\documentclass[final,12pt]{article}
\usepackage{amsfonts,color,morefloats,pslatex,a4wide}
\usepackage{amssymb,amsthm,amsmath,latexsym,pslatex,cite}
\usepackage{lineno,hyperref,mathtools}
\usepackage{pdflscape}

\newtheorem{theorem}{Theorem}[section]

\newtheorem{lemma}[theorem]{Lemma}

\theoremstyle{definition}
\newtheorem{definition}[theorem]{Definition}
\newtheorem{example}[theorem]{Example}

\theoremstyle{remark}

\newcommand{\F}{\mathbb{F}}
\newcommand{\x}{\mathbf{x}}

\newcommand{\C}{{\mathcal{C}}}

\newcommand{\GRS}{{\mathrm{GRS}}}

\newcommand{\rank}{{\rm rank}}

\makeatletter

\newcommand{\Rmnum}[1]{\expandafter\@slowromancap\romannumeral #1@}
\makeatother

\begin{document}
	\begin{sloppypar}
		\title{\Large{New MDS codes of non-GRS type and NMDS codes  }
			\thanks{The authors are with School of Mathematics, Hefei University of Technology, Hefei, 230601, China (email: yujiezhimath@163.com, zhushixinmath@hfut.edu.cn).}
			\thanks{This research is supported by the National Natural Science Foundation of China (Grant Nos. U21A20428 and 12171134).}}
		\author{Yujie Zhi, Shixin Zhu\thanks{Corresponding author}}
		\maketitle
		
		\begin{abstract} 
			Maximum distance separable (MDS) and near maximum distance separable (NMDS) codes have been widely used in various fields such as communication systems, data storage, and quantum codes due to their algebraic properties and excellent error-correcting capabilities. This paper focuses on a specific class of linear codes and establishes necessary and sufficient conditions for them to be MDS or NMDS. Additionally, we employ the well-known Schur method to demonstrate that they are non-equivalent to generalized Reed-Solomon codes.
		\end{abstract}
		{\bf Mathematics Subject Classificatio:} 94B05 12E10

		\section{Introduction}\label{sec.introduction}
		
		Let $\F_q$ be the finite field of size $q=p^m$, where $p$ is a prime and $m$ is a positive integer. Let $\F_q^*=\F_q\setminus \{0\}$ be the multiplicative group of $\F_q$.
		An $[n,k,d]_q$ {\em linear code} $\C$ is a $k$-dimensional subspace of $\F_q^n$ with a minimum Hamming distance $d$. The {\em Euclidean dual code} of an $[n,k,d]_q$ linear code $\C$ is given by 	
		\begin{align}
			\C^{\perp}=\left\{\mathbf{y}\in \F_q^n: \langle \mathbf{x},\mathbf{y} \rangle=0, \forall~ \mathbf{x}\in \C \right\},
		\end{align}
		where $\langle \mathbf{x}, \mathbf{y} \rangle=\sum_{i=1}^{n}x_iy_i$  
		for any two vectors $\mathbf{x}=(x_1,x_2,\cdots,x_n)$ and $\mathbf{y}=(y_1,y_2,\cdots,y_n) \in \F_q^n$.
		
		For any $[n,k,d]_q$ linear code $\C$, the parameters $n$, $k$ and $d$ are subject to the well-known {\em Singleton bound}: $d\leq n-k+1$ \cite{HP2003}. Define the {\em Singleton defect} of $\C$ by $S(\C)=n-k+1-d$. 
		If $S(\C)=0$, we call  $\C$ is a {\em maximum distance separable (MDS)} code and if $S(\C)=1$, we call $\C$ an {\em almost MDS (AMDS)} code. 
		Moreover, $\C$ is called a {\em near MDS (NMDS)} code if both $\C$ and $\C^{\perp}$ are AMDS codes.
		Let ${\bf A}=\{\alpha_1,\alpha_{2},\cdots,\alpha_n\}$ is a subset of $\F_{q}$. We say that ${\bf A}$ is an {\em $(n,t,\delta)$-set} in $\F_q$ if there exists an element $\delta\in \F_q$ such that no $t$ elements of ${\bf A}$ sum to $\delta$;
		${\bf A}$ contains a {\em t-zero-sum subset} if there exist $t$ elements of ${\bf A}$ sum to zero;
		${\bf A}$ is {\em t-zero-sum free} if ${\bf A}$ contains no zero-sum subset of size $t$.

		Both MDS codes and NMDS codes have attracted lots of attention due to their important applications in various fields, including distributed storage systems \cite{CHL2011}, combinatorial designs \cite{DT2020}, finite geometry \cite{DL1994}, random error channels \cite{ST2013}, informed source and index coding problems \cite{TR2018}, and secret sharing scheme \cite{SV2018,ZWXLQY2009}.

		Let $k$ and $n$ be two positive integers with $1\leq k\leq n\leq q$. Let ${\bf v}=(v_1,v_2,\cdots,v_n)\in (\F_q^*)^n$ and ${\bf A}=\{\alpha_1,\alpha_2,\cdots,\alpha_n\}\subseteq \F_q$. 
		Then an $[n,k,n-k+1]_{q}$ {\em generalized Reed-Solomon (GRS) code} associated with $\mathbf{A}$ and $\mathbf{v}$ is defined by  
		\begin{align}\label{eq.def GRS}
			\begin{split}
				\GRS_k(\mathbf{A},\mathbf{v}) =  \{(v_1f(\alpha_1),v_2f(\alpha_2),\cdots,v_nf(\alpha_n))\mid f(x)\in \mathbb{F}_{q}[x]_{k}\}.
			\end{split}
		\end{align} 
		From \cite{HP2003}, a generator matrix of $\GRS_k(\mathbf{A},\mathbf{v})$ has the following form    
		\begin{equation}\label{eq.GRS.generator matrix}
			G_{\GRS_k}=\left(\begin{array}{cccc}
				v_1 & v_2 & \cdots & v_n \\
				v_1\alpha_1 & v_2\alpha_2 & \cdots & v_n\alpha_n \\
				v_1\alpha^2_1 & v_2\alpha^2_2 & \cdots & v_n\alpha^2_n \\
				\vdots & \vdots & \ddots & \vdots \\
				v_1\alpha^{k-1}_1 & v_2\alpha^{k-1}_2 & \cdots & v_n\alpha^{k-1}_n \\
			\end{array}\right).
		\end{equation}
		Thus, using the generator matrix of $\GRS_k(\mathbf{A},\mathbf{v})$ provides a way to construct MDS codes that are equivalent to GRS codes.
		However, it is still very interesting to construct non-GRS MDS codes,	because the non-GRS properties make them resistant to Wieschebrink and Sidelnikov-Shestakov attacks \cite{LR2020,BBPR2018}. There are many constructions of non-GRS  MDS codes and NMDS codes, including: adding columns in the generator matrices of GRS codes\cite{RL1989,HL2023}, adding twists in GRS codes (see \cite{BPR2022,GLLS2023,ZZT2022,HYNL2021,SYLH2022,SYS2023,W2021,ZLA2024,WHL2021,LL2021,LCEL2022,CW2023,GZ2023} and the references therein), removing a row from the generator matrices of GRS codes\cite{HZ2023}, adding a column and removing a row from the generator matrices of GRS codes\cite{LZ2024}, and so forth.
		
		Specifically, Roth and Lempel in \cite{RL1989} introduced a construction of non-GRS MDS codes by adding two columns in the generator matrices of GRS codes. Their codes $\C_1$ over $\F_q$ have the following generator matrices of the form: 
		\begin{equation}\label{eq.RL generator matrix} G_1=
			\left( \begin{array}{cccccc}
				1 & 1& \cdots &  1 & 0 &0\\
				\alpha_1& \alpha_2& \cdots &  \alpha_{n} & 0 &0\\
				\vdots& \vdots& \ddots & \vdots & \vdots &\vdots\\
				\alpha_1^{k-3}& \alpha_2^{k-3}& \cdots &  \alpha_{n}^{k-3} & 0 &0\\
				\alpha_1^{k-2}& \alpha_2^{k-2}& \cdots &  \alpha_{n}^{k-2} & 0 &1\\
				\alpha_1^{k-1}& \alpha_2^{k-1}& \cdots &  \alpha_{n}^{k-1} & 1 &\delta\\
			\end{array} \right),
		\end{equation} 	
		where $\alpha_{1},\alpha_{2}\cdots,\alpha_{n}$ are distinct elements in $\F_q$,  $ \delta\in \F_q$, and $4\le k + 1 \le n \le q$. They proved that $\C_1$ is a non-GRS MDS code if and only if the set $\{\alpha_1,\alpha_2,\cdots,\alpha_{n}\}$ is an $(n,k-1,\delta)$-set in $\F_{q}$. Han and Fan in \cite{HF2023} further obtained NMDS codes when $\delta=0$. They proved that $\C_1$ is an NMDS code if and only if $\{\alpha_{1},\alpha_{2},\cdots,\alpha_{n}\}$ contains a $(k-1)$-zero-sum subset in $\F_{q}$.

		Based on Roth-Lempel codes, Wu et~al. in \cite{WHLD2024} gave a construction of non-GRS MDS codes by adding a column in the generator matrices $G_1$. Their codes $\C_2$ over $\F_q$ have the following generator matrices of the form: 
		\begin{equation}\label{eq.Wu generator matrix} G_2=
			\left( \begin{array}{cccccccccccc}
				1 & 1& \cdots &  1 & 0 & 0 & 0\\
				\alpha_1& \alpha_2& \cdots &  \alpha_{n} & 0 & 0 & 0\\
				\vdots& \vdots& \ddots & \vdots & \vdots & \vdots & \vdots\\
				\alpha_1^{k-3}& \alpha_2^{k-3}& \cdots &  \alpha_{n}^{k-3} & 0 & 0 & 1\\
				\alpha_1^{k-2}& \alpha_2^{k-2}& \cdots &  \alpha_{n}^{k-2} & 0 & 1 & \tau\\
				\alpha_1^{k-1}& \alpha_2^{k-1}& \cdots &  \alpha_{n}^{k-1} & 1 & \delta & \pi\\
			\end{array} \right),
		\end{equation} 	
		where $\alpha_{1},\alpha_{2}\cdots,\alpha_{n}$ are distinct elements in $\F_q$ and $ \delta, \tau, \pi \in \F_q$, and $4\le k + 1 \le n \le q$. They proved that $\C_2$ is a non-GRS MDS code if and only if the set $\{\alpha_1, \cdots, \alpha_n\}$ is an $(n,k-1,\delta)$-set and an $(n,k-2,\tau)$-set in $\F_q$, and for any subsets $I,J\subseteq \{1,2, \cdots, n\}$ with size $k-1$ and $k-2$, respectively, $\sum\limits_{i\neq j\in I}\alpha_i\alpha_j+\pi\neq \tau (\sum\limits_{i\in I} \alpha_i) 
		\mbox{   }\mbox{ and } \mbox{   }
		\pi+\delta(\sum\limits_{i\in J}\alpha_i)\neq \tau\delta+\sum\limits_{i\in J}\alpha_i^2+\sum\limits_{ i\neq j\in J}\alpha_i\alpha_j$ hold.
		They also determined that $\C_2$ is an NMDS code if and only if there exists a subset $I\subseteq \{1,2, \cdots, n\}$ with size $k-1$ such that $\sum\limits_{i\neq j\in I}\alpha_i\alpha_j+\pi= \tau (\sum\limits_{i\in I} \alpha_i)$ or  there exists a subset $J\subseteq \{1,2, \cdots, n\}$ with size $k-2$ such that $\pi+\delta(\sum\limits_{i\in J}\alpha_i)= \tau\delta+\sum\limits_{i\in J}\alpha_i^2+\sum\limits_{ i\neq j\in J}\alpha_i\alpha_j.$
		
		Han and Zhang in \cite{HZ2023} constructed MDS codes and NMDS codes by removing a row in the generator matrices of GRS codes.
		Their codes $\C_3$ over $\F_q$ have the following generator matrices of the form: 
		\begin{equation}\label{eq.HZ generator matrix} G_3=
			\left( \begin{array}{cccccc}
				1 & 1& \cdots &  1  \\
				\alpha_1& \alpha_2& \cdots &  \alpha_{n}\\
				\vdots& \vdots& \ddots & \vdots \\
				\alpha_1^{k-2}& \alpha_2^{k-2}& \cdots &  \alpha_{n}^{k-2} \\
				\alpha_1^{k}& \alpha_2^{k}& \cdots &  \alpha_{n}^{k} \\
			\end{array} \right),
		\end{equation} 	
		where $\alpha_{1},\alpha_{2},\cdots,\alpha_{n}$ are distinct elements in $\F_q$ and $3\le k \le n \le q$. They proved that $\C_3$ is an MDS code if and only if $\{\alpha_{1},\alpha_{2},\cdots,\alpha_{n}\}$ is $k$-zero-sum free and $\C_3$ is an NMDS code if and only if $\{\alpha_{1},\alpha_{2},\cdots,\alpha_{n}\}$ contains a zero-sum subset of size $k$ in $\F_{q}$.

		In \cite{LZ2024}, the authors introduced another construction of  MDS codes and NMDS codes by adding one column in the generator matrices $G_3$ in \cite{HZ2023}.
		Their codes $\C_4$ over $\F_q$ have the following generator matrices of the form: 
		\begin{equation}\label{eq.LZ generator matrix} G_4=
			\left( \begin{array}{cccccc}
				1 & 1& \cdots &  1 & 0 \\
				\alpha_1& \alpha_2& \cdots &  \alpha_{n} & 0 \\
				\vdots& \vdots& \ddots & \vdots & \vdots \\
				\alpha_1^{k-2}& \alpha_2^{k-2}& \cdots &  \alpha_{n}^{k-2} & 0 \\
				\alpha_1^{k}& \alpha_2^{k}& \cdots &  \alpha_{n}^{k} & 1 \\
			\end{array} \right),
		\end{equation} 	
		where $\alpha_{1},\alpha_{2},\cdots,\alpha_{n}$ are distinct elements in $\F_q$ and $3\le k \le n \le q$. They proved that $\C_4$ is a non-GRS MDS code if and only if the set $\{\alpha_1,\alpha_2,\cdots,\alpha_{n}\}$ is $k$-zero-sum free in $\F_{q}$ and $\C_4$ is an NMDS code if and only if the set $\{\alpha_1,\alpha_2,\cdots,\alpha_{n}\}$ contains a $k$-zero-sum subset in $\F_{q}$.  
		
		Inspired by above constructions, in order to obtain more non-GRS MDS codes and NMDS codes, it is nature to consider the following two classes of codes:
		\begin{itemize}
			\item As \cite{HZ2023} and \cite{LZ2024} do, we delete the penultimate row of $G_1$ to get a class of code, it is not hard to see that there exist two columns of the deleted matrices are linearly dependent. Therefore, it is uninteresting.
			\item As \cite{RL1989} and \cite{WHLD2024} do, we add a column in the generator matrices $G_4$ (i.e., the linear code $\C_{k}({\bf A},{\bf v})$ (see Definition \ref{def.codes})).
			Our main contributions can be summarized as follows: 
			\begin{itemize}
				\item We prove $\C_k({\bf A},{\bf v})$ is an extended code of $\C_4$ generated by $G_4$ in Theorem \ref{th.extended code}.
				
				\item We determine the parameters and a parity-check matrix of $\C_k({\bf A},{\bf v})$ in Theorems \ref{th.length and dimension} and \ref{th.parity-check matrix}.
				
				\item We determine sufficient and necessary conditions for $\C_k({\bf A},{\bf v})$ to be MDS in Theorem \ref{th.MDS condition}.  Based on the Schur method, we further show that such codes are non-GRS in Theorem \ref{th.non-GRS}. 
				
				\item We characterize sufficient and necessary conditions for $\C_k({\bf A},{\bf v})$ to be NMDS in Theorem \ref{th.NMDS condition}. 
			\end{itemize}		
			
		\end{itemize}

		This paper is organized as follows. 
		In Section \ref{sec2.preliminaries}, we recall some knowledge on the Schur product, a kind of extended codes of linear codes and some useful results. 
		In Section \ref{sec3}, we prove that $\C_{k}({\bf A}, {\bf v})$ is an extended code of $\C_4$ and study the parameters and a parity-check matrix of $\C_{k}({\bf A}, {\bf v})$.
		In Section \ref{sec4}, we determine sufficient and necessary conditions for these codes to be non-GRS MDS codes. 
		In Section \ref{sec5}, we determine sufficient and necessary conditions for these codes to be NMDS codes. 
		In Section \ref{sec.concluding remarks}, we conclude our main contributions in this paper.

		\section{Preliminaries}\label{sec2.preliminaries}
		\subsection{The Schur product}
		
		In this subsection, we will recall some useful results on the Schur product.
		
		\begin{definition}\label{def.Schur product}
			For any two vectors ${\bf x}=(x_1,x_2,\cdots,x_n) \mbox{ and } {\bf y}=(y_1,y_2\cdots,y_n) \in \F_q^n$, we define the Schur product ${\bf x} \star {\bf y}\in \F_q^n$ by
			$${\bf x} \star {\bf y}=(x_1\cdot y_1,x_2\cdot y_2,\cdots,x_n\cdot y_n).$$
		\end{definition}
		
		\begin{definition}\label{def.Schur product codes}
			For given two $[n,k]_q$ linear codes $\C_1$ and $\C_2$, the Schur product $\C_1 \star \C_2$ is defined by $${\C_1} \star {\C_2}=\{{\bf c_1} \star {\bf c_2}:  {\bf c_1}\in \C_1 \mbox{ and } {\bf c_2}\in \C_2\}.$$ 
			
			If $\C_1=\C_2=\C$, we further denote ${\C_1} \star {\C_2}=\C^2$ and call it the Schur square of $\C$.
		\end{definition}
		
		The following two lemmas give the Schur square of a GRS code. 
		
		\begin{lemma}{\rm (\!\!\cite[Proposition 10]{MMP2013},\cite[Proposition 2.1 (1)]{ZLT2024},\cite[Lemma 3.3]{HL2023})}\label{lem.GRS square dimension}
			If $\C$ is an $[n+1,k]_q$ GRS code with $3\leq k<\frac{n+2}{2}$, then $$\dim(\C^2)=2k-1.$$ 
		\end{lemma}
		
		\begin{lemma}{\rm (\!\!\cite[Lemma 3.3]{HL2023})}\label{lem.GRS square distance}
			If $\C$ is an $[n+1,k]_q$ GRS code with $3\leq k<\frac{n+2}{2}$, then $$d(\C^2)\geq 2.$$ 
		\end{lemma}
		
		\subsection{A kind of extended codes of linear codes}
		
		In this subsection, we will introduce a kind of extended codes of linear codes proposed by Sun et al. in \cite{SDC2023}.
		
		\begin{definition}\label{def.extended codes}
			Let ${\bf w}=(w_1,w_2,\cdots,w_n)\in \F_q^n$ be any nonzero vector. Any given $[n,k,d]$ linear code $\C$ can be extended into an $[n+1,k,\overline{d}]$ code $\overline{\C}({\bf w})$ over $\F_{q}$ as follows:
			\begin{equation}\label{eq.extended code}
				\overline{\C}({\bf w})=\bigg\{  (c_1, \cdots, c_n,c_{n+1}): (c_1, \cdots, c_n)\in \mathcal C,c_{n+1}=\sum_{i=1}^nw_ic_i\bigg\},
			\end{equation}
			where $\overline{d} = d$ or $\overline{d} = d +1$.  
		\end{definition}

		\begin{lemma}\label{lem.extended code generator and parity-check matrix}{\rm (\!\!\cite{SDC2023})}
			Let $\C$ be an $[n,k,d]$ linear code over $\F_q$ and ${\bf w} \in \F_q^n$. If $\C $ has a generator matrix $G$ and a parity-check matrix $H$, then the generator and parity-check matrices for the extended code $\overline{\C}({\bf w})$ defined in \eqref{eq.extended code} are $$(G~ G{\bf w}^T) \mbox{ and } 
			\left(\begin{array}{cccc} H & {\bf 0}^T\\ {\bf w} & -1
			\end{array}\right),$$
			where ${\bf w}^T$ denotes the transpose of ${\bf w}$.
			
		\end{lemma}

		\subsection{Some auxiliary results} 
		
		This subsection reviews some useful results for later use. 
		
		\begin{lemma}{\rm(\!\!\cite[Page 466]{H1929}, \cite[Lemma 17]{LDMTT2022})}
			For any integer $n\geq 3$, it holds that 
			\begin{align}\label{eq.det1}
				\det\left(\begin{array}{ccccc}
					1 & 1 & \cdots & 1 & 1 \\ 
					\alpha_1 & \alpha_2 & \cdots & \alpha_{n-1} & \alpha_n \\
					\alpha_1^2 & \alpha_2^2 & \cdots & \alpha_{n-1}^2 & \alpha_n^2 \\
					\vdots & \vdots & \ddots & \vdots & \vdots \\
					\alpha_1^{n-2} & \alpha_2^{n-2} & \cdots & \alpha_{n-1}^{n-2} & \alpha_n^{n-2} \\
					\alpha_1^n & \alpha_2^n & \cdots & \alpha_{n-1}^{n} & \alpha_n^n 
				\end{array}\right)=\sum_{i=1}^{n}\alpha_i\prod_{1\leq i<j\leq n}(\alpha_j-\alpha_i).               
			\end{align}
		\end{lemma}

		\begin{lemma}{\rm(\!\!\cite[Lemma 2]{FX2024})}
			For any integer $n\geq 3$, it holds that
			\begin{align}\label{eq.det2}
				\det\left(\begin{array}{ccccc}
					1 & 1 & \cdots & 1 & 1 \\ 
					\alpha_1 & \alpha_2 & \cdots & \alpha_{n-1} & \alpha_n \\
					\alpha_1^2 & \alpha_2^2 & \cdots & \alpha_{n-1}^2 & \alpha_n^2 \\
					\vdots & \vdots & \ddots & \vdots & \vdots \\
					\alpha_1^{n-2} & \alpha_2^{n-2} & \cdots & \alpha_{n-1}^{n-2} & \alpha_n^{n-2} \\
					\alpha_1^{n+1 }& \alpha_2^{n+1} & \cdots & \alpha_{n-1}^{n+1} & \alpha_n^{n+1} 
				\end{array}\right)
				=\bigg((\sum_{i=1}^{n}\alpha_i)^2-\sum_{1\leq i<j\leq n}\alpha_i\alpha_j\bigg)\prod_{1\leq i<j\leq n}(\alpha_j-\alpha_i). 		            
			\end{align}
		\end{lemma}

		\begin{lemma}\label{lem.PRA}
			Let $u_i=\prod\limits_{1\leq j\leq n \atop j\neq i}(\alpha_i-\alpha_j)^{-1}$ for $1\leq i\leq n$. For any subset ${\bf A}=\{\alpha_1,\alpha_2,\cdots,\alpha_n\}\subseteq \F_q$ with $n\geq 3$, it holds that  
			\begin{align}\label{eq.ui value}
				\sum_{i=1}^{n}\alpha_i^{\ell}u_i=\left\{\begin{array}{ll}
					0, & 0\leq \ell\leq n-2, \\
					1, & \ell=n-1, \\
					\sum\limits_{i=1}^n \alpha_i, & \ell=n, \\
					(\sum\limits_{i=1}^{n}\alpha_i)^2-\sum\limits_{1\leq i<j\leq n}\alpha_i\alpha_j, &\ell=n+1.
				\end{array}\right.
			\end{align}
		\end{lemma}
		
		\begin{proof}
			When $0\leq \ell\leq n-2$, $\ell=n-1$ and $\ell=n$, the desired results have been documented in \cite[Lemma 5]{LXW2008}, \cite[Lemma \Rmnum{2}.1]{FLL2020} and \cite[Lemma 2.8] {LZ2024}, respectively. 
			When $\ell=n+1$, combining Equation \eqref{eq.det2} and the proof given in \cite[Lemma 5]{LXW2008}, we get the desired results.
		\end{proof}
		
		The following lemmas can be used to determine whether a given linear code is an MDS code or NMDS code.
		
		\begin{lemma}{\rm (\!\!\cite[Lemma 7.3]{B2015})}\label{lem.MDS} 
			Let $\C$ be an $[n,k,d]_q$ linear code with a generator matrix $G$. 
			Then $\C$ is MDS if and only if any $k$ columns of $G$ are linearly independent.	 
		\end{lemma}	

	\begin{lemma}{\rm (\!\!\cite[Lemma 3.1]{DL1994})}\label{lem.NMDS2} 
				Let $\C$ be an $[n,k,d]_q$ linear code with a generator matrix $G$. 
			Then $\C$ is NMDS if and only if 
			\begin{itemize}
				\item any $k-1$ columns of $G$ are linearly independent;
				\item there exist $k$ linearly dependent columns in $G$;
				\item any $k+1$ columns of $G$ are of full rank.
			\end{itemize}
		\end{lemma}

		\section{The class of linear codes $\C_k(\bf{A},{\bf v})$}\label{sec3}
		
		In this section, we give the definition of $\C_k(\bf{A},{\bf v})$ and examine their parameters and parity-check matrices. We also prove $\C_k(\bf{A},{\bf v})$ is an extended code of $\C_4$ in\cite{LZ2024}. 
		
		The linear codes $\C_k({\bf A},{\bf v})$ are defined as follow.
		\begin{definition}\label{def.codes}
			Let $n$ and $k$ be two positive integers satisfying $3\leq k\leq n\leq q$.
			Let ${\bf A}=\{\alpha_1,\alpha_2,\cdots,\alpha_{n}\}\subseteq \F_q$, ${\bf v}=(v_1,v_2,\cdots,v_{n})\in (\F_q^*)^n$ and $\delta\in \F_{q}$.  
			Denote by $\C_k({\bf A},{\bf v})$ the $q$-ary linear code generated by the matrix 
			\begin{equation}\label{eq.generator matrix}
				G_k=\begin{pmatrix}    
					v_1 & v_2 & \cdots & v_{n} & 0 & 0\\ 
					v_1\alpha_1 & v_2\alpha_2 & \cdots & v_{n}\alpha_{n} & 0 & 0\\
					v_1\alpha_1^2 & v_2\alpha_2^2 & \cdots & v_{n}\alpha_{n}^2 & 0 & 0\\
					\vdots & \vdots & \ddots & \vdots & \vdots & \vdots\\
					v_1\alpha_1^{k-2} & v_2\alpha_2^{k-2} & \cdots & v_{n}\alpha_{n}^{k-2} & 0 & 1\\
					v_1\alpha_1^k & v_2\alpha_2^k & \cdots & v_{n}\alpha_{n}^k & 1 & \delta\\
				\end{pmatrix}.      
			\end{equation}        
			Let $V_k=\left\{f(x)=\sum_{i=0}^{k-2}f_ix^i+f_{k}x^{k}:~f_i\in \F_q, i\in \{0,1,\cdots,k-2,k\}\right\}$.
			Then $\C_k({\bf A},{\bf v})$ can be further expressed as  
			\begin{align*}
				\C_k({\bf A},{\bf v})=\left\{ \left(v_1f(\alpha_1),v_2f(\alpha_2),\cdots,v_nf(\alpha_n),f_{k},f_{k-2}+\delta f_{k}\right): ~f(x)\in V_k \right\},  
			\end{align*}
			where $f_{k}$ and $f_{k-2}$ are the coefficients of $x^k$ and $x^{k-2}$ in $f(x)$, respectively.     
		\end{definition}
		
		\subsection{An extended code of linear code $\C_4$}
		In this subsection, we will give a theorem for the relationship between $\C_k({\bf A},{\bf v})$ and $\C_4$.
		
		\begin{theorem}\label{th.extended code}
			The $q$-ary linear code $\C_k({\bf A},{\bf v})$ generated by $G_k$ given in Equation (\ref{eq.generator matrix}) is an extended code of the code $\C_4$ generated by $G_4$ in Equation \eqref{eq.LZ generator matrix}, denoted by $$\C_k({\bf A},{\bf v})=\overline{\C_4}(\bf w),$$ where ${\bf w}=(w_1, w_2, \cdots, w_{n+1})\in \F_q^{n+1}$ and for $1\le i\le n,$
			$$w_i=\alpha_i^{n-k+1}\prod_{1\leq j\leq n\atop j\neq i}(\alpha_i-\alpha_j)^{-1}$$
			and
			$$w_{n+1}=\delta-(\sum_{i=1}^n \alpha_i)^2+\sum_{1\le i<j\le n}\alpha_i\alpha_j.$$
		\end{theorem}
		
		\begin{proof}  
			By Lemma \ref{lem.extended code generator and parity-check matrix}, it suffices  to check that $G_{4}{\bf w}^T=(0, 0,\cdots, 0,1, \delta)^T$. 
			
			According to the Cramer Rule, the system of the equations
			\begin{equation}\label{eq.equations system}
				\left(\begin{array}{ccccccc} 
					v_1 & v_2 & \cdots  & v_n\\ 
					v_1\alpha_1 & v_2\alpha_2 & \cdots & v_n\alpha_{n}\\
					\vdots & \vdots & \ddots & \vdots \\ 
					v_1\alpha_1^{n-1} & v_2\alpha_2^{n-1} & \cdots & v_n\alpha_{n}^{n-1}
				\end{array}\right)\left(\begin{array}{c}x_1\\x_2\\ \vdots
					\\ x_n\end{array}\right)=\left(\begin{array}{c}0\\
					0\\ \vdots\\  1 \end{array}\right)
			\end{equation}		
			has a nonzero solution $(u_1,u_2,\cdots, u_{n})^T$, where 
			$$u_i=\prod\limits_{1\leq j\leq n \atop j\neq i}(\alpha_i-\alpha_j)^{-1},\mbox{  } i=1,2,\cdots,n.$$  
			
			From Equation \eqref{eq.equations system}, we only need to check that  $$(\alpha_1^{k},\alpha_2^{k}, \cdots, \alpha_n^{k},1){\bf w}^T=\sum_{i=1}^nu_i\alpha_i^{n+1}+w_{n+1}=\delta.$$
			
			Then, we have $$w_{n+1}=\delta-(\sum_{i=1}^n \alpha_i)^2+\sum_{1\le i<j\le n}\alpha_i\alpha_j.$$
			
			This completes the proof.
		\end{proof}

		\subsection{Parameters and a parity-check matrix of $\C_k({\bf A},{\bf v})$} 
		
		In this subsection, we determine parameters and a parity-check matrix of $\C_k({\bf A},{\bf v})$.
		
		\begin{theorem}\label{th.length and dimension} 
			Let $\C_k({\bf A},{\bf v})$ be the $q$-ary linear code generated by $G_k$ given in Equation (\ref{eq.generator matrix}). 
			Then $\C_k({\bf A},{\bf v})$ has parameters $[n+2,k]_q$. 
		\end{theorem}
		\begin{proof} 
			From Definition \ref{def.extended codes}, combining Theorem \ref{th.extended code} and {\rm \cite[Theorem 3.2]{LZ2024}}, we immediately have that $\C_k({\bf A},{\bf v})$ has parameters $[n+2,k]_q$.
		\end{proof}	
		
		\begin{theorem}\label{th.parity-check matrix}    
			Let $\C_k({\bf A},{\bf v})$ be the $q$-ary linear code generated by $G_k$ given in Equation (\ref{eq.generator matrix}). 
			Suppose that $v_i'=u_iv_i^{-1}, u_i=\prod\limits_{1\leq j\leq n \atop j\neq i}(\alpha_i-\alpha_j)^{-1}$, for any $1\leq i\leq n$. 
			Then the following statements hold. 
			\begin{enumerate}
				\item [\rm (1)]The $(n-k+2)\times (n+2)$ matrix 
				\begin{equation}\label{eq.parity-check matrix}
					H_{n-k+2}=\begin{pmatrix}    
						v_1' & v_2' & \cdots & v_n' & 0 & 0 \\ 
						v_1'\alpha_1 & v_2'\alpha_2 & \cdots & v_n'\alpha_{n} & 0 & 0 \\
						\vdots & \vdots & \ddots & \vdots & \vdots & \vdots \\
						v_1'\alpha_1^{n-k-1} & v_2'\alpha_2^{n-k-1} & \cdots & v_n'\alpha_{n}^{n-k-1} & -1 & 0\\
						v_1'\alpha_1^{n-k} & v_2'\alpha_2^{n-k} & \cdots & v_n'\alpha_{n}^{n-k} & a  & 0\\
						v_1'\alpha_1^{n-k+1} & v_2'\alpha_2^{n-k+1} & \cdots & v_n'\alpha_{n}^{n-k+1}  & b & -1\\
					\end{pmatrix},     
				\end{equation}
				where $a=-\sum\limits_{i=1}^{n}\alpha_i$ and $b=\delta+\sum\limits_{1\leq i<j\leq n}\alpha_i\alpha_j-(\sum\limits_{i=1}^{n}\alpha_i)^2$, is a parity-check matrix of $\C_k({\bf A},{\bf v})$. 
				
				\item [\rm (2)] The dual code $\C_k({\bf A},{\bf v})^{\perp}$ is 
				\begin{align*}
					\left\{\left(v_1'g(\alpha_1), v_2'g(\alpha_2), \cdots, v_n'g(\alpha_n), -g_{n-k-1}+ag_{n-k}+bg_{n-k+1}, -g_{n-k+1}\right):~g(x)\in \F_q[x]_{n-k+2}\right\}, 
				\end{align*}
				where $a=-\sum\limits_{i=1}^{n}\alpha_i$ and $b=\delta+\sum\limits_{1\leq i<j\leq n}\alpha_i\alpha_j-(\sum\limits_{i=1}^{n}\alpha_i)^2.$
			\end{enumerate}
		\end{theorem}
		\begin{proof}
			(1) On the one hand, it is easy to see that $\rank(H_{n-k+2})=n-k+2$. 
			On the other hand, we denote the $s$-th row of the generator matrix $G_k$ by ${\bf g}_s$ for any $1\leq s\leq k$ and denote the $t$-th row of the matrix $H_{n-k+2}$ by $\mathbf{h}_t$ for any $1\leq t\leq n-k+2$.  
			Then our main task becomes to prove ${\bf g}_s {\bf h}^T_t=0$ for any possible $s$ and $t$. 
			It is not difficult to calculate that  
			\begin{align*}   	
				{\bf g}_s {\bf h}^T_t=\left\{\begin{array}{ll}
					\sum\limits_{i=1}^n u_i\alpha_i^{s+t-2}, & {\rm if}~1\leq s\leq k-2~{\rm and}~1\leq t\leq n-k+2, \\ 
					\sum\limits_{i=1}^n u_i\alpha_i^{k+t-3}, & {\rm if}~s=k-1~{\rm and}~ 1\leq t\leq n-k+1, \\ 
					\sum\limits_{i=1}^n u_i\alpha_i^{n-1}-1, & {\rm if}~s=k-1~{\rm and}~t=n-k+2, \\ 
					\sum\limits_{i=1}^n u_i\alpha_i^{k+t-1}, & {\rm if}~s=k~{\rm and}~ 1\leq t\leq n-k-1, \\ 
					\sum\limits_{i=1}^n u_i\alpha_i^{n-1}-1, & {\rm if}~s=k~{\rm and}~t=n-k, \\ 
					\sum\limits_{i=1}^n u_i\alpha_i^{n}+a, & {\rm if}~s=k~{\rm and}~t=n-k+1, \\ 
					\sum\limits_{i=1}^n u_i\alpha_i^{n+1}+b-\delta, & {\rm if}~s=k~{\rm and}~t=n-k+2. 
				\end{array} \right.   
			\end{align*}
			Then the desired result (1) follows straightforward from Lemma \ref{lem.PRA}. 
			
			(2) Since a parity-check matrix of a linear code is a generator matrix of its dual code, then the expected result (2) immediately holds according to the result (1) above. 
			
			This completes the proof. 
		\end{proof}

		\section{MDS codes $\C_k({\bf A},{\bf v})$ of non-GRS type}\label{sec4}
		\subsection{When is the code $\C_k({\bf A},{\bf v})$ MDS?}
		In this subsection, our objective is to determine some necessary and sufficient conditions under which the code $\C_k({\bf A},{\bf v})$ is MDS. 
		By thoroughly investigating these conditions, our aim is to gain a comprehensive understanding of the MDS property of $\C_k({\bf A},{\bf v})$.
		
		\begin{theorem}\label{th.MDS condition}
			The code $\C_k({\bf A},{\bf v})$ generated by $G_k$ given in Equation (\ref{eq.generator matrix}) is an MDS code if and only if the following conditions hold:
			\begin{enumerate}
				\item [\rm (1)] For any subset $K$ with size $k$ of $\{1,2, \cdots, n\}$, $$\sum_{i\in K}\alpha_i\neq 0.$$
				
				\item [\rm (2)] For any subset $I$ with size $k-1$ of $\{1,2, \cdots, n\}$, $$(\sum_{i\in I}\alpha_i)^2-\sum_{ i<j\in I}\alpha_i\alpha_j\neq \delta.$$
			\end{enumerate}
			
		\end{theorem}
		
		\begin{proof} 
			The code $\C_k({\bf A}, {\bf v})$ is MDS if and only if the submatrix  consisting of any $k$ columns from
			$G_k$ in Equation (\ref{eq.generator matrix}) is nonsingular. Let ${\bf d}_1=(0,0,\cdots, 0,1)^T$  and  ${\bf d}_2=(0,0,\cdots,0,1, \delta)^T$. Then we divide the proof into the following cases:
			\begin{itemize}
				\item [] {\textbf{Case 1.}} Assume that the submatrix consisting of $k$ columns from $G_k$ does not contain ${\bf d}_1$ and  ${\bf d}_2$. Suppose that the submatrix contains any $k$ columns of the first $n$ columns from $G_k$ and note $K=\{i_1,i_2,\cdots,i_k\}\subseteq\{1,2,\cdots,n\}$. Then we need to compute the determinant of the submatrix:
				\begin{equation*}\label{eq.B_1}
					B_1=\left(\begin{array}{ccccc}
						1 & 1 & \cdots & 1\\ 
						\alpha_{i_1} & \alpha_{i_2} & \cdots & \alpha_{i_k}\\ 
						\vdots & \vdots & \ddots & \vdots \\ 
						\alpha_{i_1}^{k-3} & \alpha_{i_2}^{k-3} & \cdots & \alpha_{i_k}^{k-3}\\ 
						\alpha_{i_1}^{k-2} & \alpha_{i_2}^{k-2} & \cdots & \alpha_{i_k}^{k-2}\\
						\alpha_{i_1}^{k} & \alpha_{i_2}^{k} & \cdots & \alpha_{i_k}^{k}\\
					\end{array}\right)_{k\times k}.
				\end{equation*} 	
				Then $$\mbox{det}(B_1)=\sum_{i\in K}\alpha_i \prod_{i<j\in K}(\alpha_j-\alpha_i).$$
				Thus $\mbox{det}(B_1)\neq 0$ if and only if $\sum\limits_{i\in K} \alpha_i \neq 0.$
				Hence, the matrix is nonsigular if and only if $\sum\limits_{ i\in K}\alpha_i\neq 0.$
				
				\item [] {\textbf{Case 2.}} Assume that the submatrix consists only ${\bf d}_1$ and other $k-1$ columns from $G_k$. Suppose that the submatrix contains only ${\bf d}_1$ and any $k-1$ columns of the first $n$ columns from $G_k$ and note $I=\{i_1,i_2,\cdots,i_{k-1}\}\subseteq\{1,2,\cdots,n\}$. Then we need to compute the determinant of the submatrix:
				\begin{equation*}\label{eq.B_2}
					B_2=\left(\begin{array}{ccccc}
						1 & 1 & \cdots & 1 &0\\ 
						\alpha_{i_1} & \alpha_{i_2} & \cdots & \alpha_{i_{k-1}} & 0\\ 
						\vdots & \vdots & \ddots & \vdots & \vdots\\ 
						\alpha_{i_1}^{k-2} & \alpha_{i_2}^{k-2} & \cdots & \alpha_{i_{k-1}}^{k-2} & 0\\
						\alpha_{i_1}^{k} & \alpha_{i_2}^{k} & \cdots & \alpha_{i_{k-1}} ^{k}& 1\\
					\end{array}\right).
				\end{equation*} 	
				Then $$\mbox{det}(B_2)=\prod_{ i<j\in I}(\alpha_j-\alpha_i)\neq 0.$$
				Hence, the matrix is nonsigular.
				
				\item [] {\textbf{Case 3.}} Assume that the submatrix consists only ${\bf d}_2$ and other $k-1$ columns from $G_k$. Suppose that the submatrix contains only ${\bf d}_2$ and any $k-1$ columns of the first $n$ columns from $G_k$ and note $I=\{i_1,i_2,\cdots,i_{k-1}\}\subseteq\{1,2,\cdots,n\}$.  Then we need to compute the determinant of the submatrix:
				\begin{equation*}\label{eq.B_3}
					B_3=\left(\begin{array}{ccccc}
						1 & 1 & \cdots & 1 &0\\ 
						\alpha_{i_1} & \alpha_{i_2} & \cdots & \alpha_{i_{k-1}} & 0\\ 
						\vdots & \vdots & \ddots & \vdots & \vdots\\ 
						\alpha_{i_1}^{k-2} & \alpha_{i_2}^{k-2} & \cdots & \alpha_{i_{k-1}}^{k-2} & 1\\
						\alpha_{i_1}^{k} & \alpha_{i_2}^{k} & \cdots & \alpha_{i_{k-1}}^{k} & \delta\\
					\end{array}\right).
				\end{equation*} 	
				Then by equation \eqref{eq.det2}, we have
				\begin{eqnarray*}
					\mbox{det}(B_3)&=&(-1)^{k-1+k}\left|\begin{array}{cccccc}
						1 & 1 & \cdots & 1\\ 
						\alpha_{i_1} & \alpha_{i_2} & \cdots & \alpha_{i_{k-1}}\\ 
						\vdots & \vdots & \ddots & \vdots \\ 
						\alpha_{i_1}^{k-3} & \alpha_{i_2}^{k-3} & \cdots & \alpha_{i_{k-1}}^{k-3}\\ 
						\alpha_{i_1}^{k} &\alpha_{i_2}^{k} & \cdots & \alpha_{i_{k-1}}^{k}
					\end{array}\right| +\delta \prod_{i<j\in I}(\alpha_j-\alpha_i)\\								
					&=&-\bigg(\prod_{i<j\in I}(\alpha_j-\alpha_i)\bigg)\bigg((\sum_{i\in I}\alpha_i)^2-\sum_{i<j\in I}\alpha_i\alpha_j\bigg)+\delta \prod_{i<j\in I}(\alpha_j-\alpha_i).
				\end{eqnarray*}  	
				Hence, $$\mbox{det}(B_3)=\bigg(\delta-(\sum_{i\in I}\alpha_i)^2+\sum_{i<j\in I}\alpha_i\alpha_j\bigg)\prod_{i<j\in I}(\alpha_j-\alpha_i).$$
				Therefore, $\mbox{det}(B_3)\neq 0$ if and only if $\delta-(\sum\limits_{i\in I}\alpha_i)^2+\sum\limits_{i<j\in I}\alpha_i \alpha_j\neq 0.$
				Hence, the matrix is nonsigular if and only if $$\delta-(\sum_{i\in I}\alpha_i)^2+\sum_{i<j\in I}\alpha_i \alpha_j\neq 0.$$
				
				\item [] {\textbf{Case 4.}} Assume that the submatrix consists ${\bf d}_1$, ${\bf d}_2$ and other $k-2$ columns from $G_k$. Suppose that the submatrix contains ${\bf d}_1$, ${\bf d}_2$ and any $k-2$ columns of the first $n$ columns from $G_k$ and note $J=\{i_1,i_2,\cdots,i_{k-2}\}\subseteq\{1,2,\cdots,n\}$. Then we need to compute the determinant of the submatrix:
				\begin{equation*}\label{eq.B_4}
					B_4=\left(\begin{array}{ccccccccc}
						1  & 1 & \cdots & 1 & 0 & 0\\ 
						\alpha_{i_1} & \alpha_{i_2} & \cdots & \alpha_{i_{k-2}} & 0 & 0\\ 
						\vdots  & \vdots & \ddots & \vdots & \vdots & \vdots\\ 
						\alpha_{i_1}^{k-2} & \alpha_{i_2}^{k-2} & \cdots & \alpha_{i_{k-2}}^{k-2} & 0 & 1\\
						\alpha_{i_1}^{k} & \alpha_{i_2}^{k}& \cdots & \alpha_{i_{k-2}}^{k} & 1 & \delta\\
					\end{array}\right).
				\end{equation*} 	
				Then 
				\begin{eqnarray*}
					\mbox{det}(B_4)&=&(-1)^{k-1+k}\left|\begin{array}{ccccccccc}
						1  & 1 & \cdots & 1 & 0\\ 
						\alpha_{i_1} & \alpha_{i_2} & \cdots & \alpha_{i_{k-2}}  & 0\\ 
						\vdots & \vdots & \ddots & \vdots &\vdots\\ 
						\alpha_{i_1}^{k-3} & \alpha_{i_2}^{k-3} & \cdots & \alpha_{i_{k-2}}^{k-3} & 0 \\
						\alpha_{i_1}^{k-2}  & \alpha_{i_2}^{k-2} & \cdots & \alpha_{i_{k-2}}^{k-2} & 1\\
					\end{array}\right| \\
					&=&-\prod_{i<j\in J}(\alpha_j-\alpha_i)	\neq 0.								
				\end{eqnarray*}
				Hence, the matrix is nonsigular. 
			\end{itemize}
			
			This completes the proof.
		\end{proof}	
		
		\begin{example}   
			Let $q=4$ and $k=n=3$. Let $\omega$ be a primitive element of the finite filed ${\rm GF(4)}$, where $\omega$ satisfied $\omega^2+\omega+1=0$. Then ${\rm GF(4)}=\{0,1,\omega,1+\omega\}$. Let ${\bf A}=\{0,1,\omega\}$ and $\delta=\omega$, we can easily check that ${\bf A}$ and $\delta$ satisfy Theorem \ref{th.MDS condition}.
			Then we have $\C_k({\bf A},{\bf v})$ is an MDS code with parameters $[5,3,3]$ over ${\rm GF(4)}$ and its generator matrix is 
			\begin{equation*}\left(\begin{array}{cccccc} 
					1 & 1 & 1 & 0 & 0\\ 
					0 & 1 & \omega & 0 & 1\\ 
					0 & 1 & 1 & 1 & \omega\\ 
				\end{array}\right).
			\end{equation*} 
			Verified by Magma\cite{Magma}, these results are true.
		\end{example}
		
		\begin{example}   
			Let $q=8, k=3, n=4$. Let $\gamma$ be a primitive element of the finite filed ${\rm GF(8)}$, where $\gamma$ satisfied $\gamma^3+\gamma+1=0$. Then ${\rm GF(8)}=\{0,1,\gamma,\gamma^2,1+\gamma,\gamma+\gamma^2,1+\gamma+\gamma^2,1+\gamma^2\}$. Let ${\bf A}=\{1,\gamma,\gamma^2,\gamma^5\}$ and $\delta=\gamma^4$, we can easily check that ${\bf A}$ and $\delta$ satisfy Theorem \ref{th.MDS condition}.
			Then we have $\C_k({\bf A},{\bf v})$ is an MDS code with parameters $[6,3,4]$ over ${\rm GF(5)}$ and its generator matrix is 
			\begin{equation*}\left(\begin{array}{cccccc} 
					1 & 1 & 1 & 1 & 0 & 0\\ 
					1 & \gamma & \gamma^2 & \gamma^5 & 0 & 1\\ 
					1 & \gamma^3 & \gamma & \gamma & 1 & \gamma^4\\
				\end{array}\right).
			\end{equation*}
			Verified by Magma\cite{Magma}, these results are true.
		\end{example}

		\subsection{Non-GRS property of $\C_k({\bf A},{\bf v})$}
		
		In this subsection, we determine the non-GRS property of $\C_k({\bf A},{\bf v})$.
		
		\begin{theorem}\label{th.non-GRS}
			Let $\C_k({\bf A}, {\bf v})$ be the $q$-ary linear code generated by $G_k$ given in Equation (\ref{eq.generator matrix}). 
			Then $\C_k({\bf A}, {\bf v})$ is non-GRS. 
		\end{theorem}
		\begin{proof}
			Up to equivalence, it is sufficient to consider the case where ${\bf v}={\bf 1}$. 
			We have two cases.
			\begin{enumerate}
				\item [] {\textbf{Case 1.}} If $3\leq k<\frac{n+2}{2}$, we need to compute the dimension of $\C_k({\bf A}, {\bf 1})^2$ by Lemma \ref{lem.GRS square dimension}.	
				
				Denote ${\bf a}^z=(\alpha_1^z, \alpha_2^z,\cdots,\alpha_n^z)$, for any nonnegative integer $z$, and $\C_{k}({\bf A}, {\bf 1})$=$\langle({\bf a}^i,0,0),$ $ ({\bf a}^{k-2},0,1), ({\bf a}^k,1,\delta)\rangle$ 
				for integer $0\leq i\leq k-3$. 
				By definitions, we have 
				\begin{align}\label{eq. square}
					\begin{split}
						 \C_k({\bf A},{\bf 1})^2= & \C_k({\bf A},{\bf 1})\star \C_k({\bf A},{\bf 1}) \\ 
						= & \langle ({\bf a}^i,0,0)\star ({\bf a}^j,0,0), ({\bf a}^i,0,0)\star ({\bf a}^{k-2},0,1), ({\bf a}^i,0,0)\star ({\bf a}^k,1,\delta), \\
						& ({\bf a}^{k-2},0,1)\star ({\bf a}^j,0,0), ({\bf a}^{k-2},0,1)\star ({\bf a}^{k-2},0,1), ({\bf a}^{k-2},0,1)\star ({\bf a}^{k},1,\delta), \\ 
						&({\bf a}^{k},1,\delta)\star ({\bf a}^j,0,0), ({\bf a}^{k},1,\delta)\star ({\bf a}^{k-2},0,1), ({\bf a}^{k},1,\delta)\star ({\bf a}^{k},1,\delta):~0\leq i,j\leq k-3 \rangle \\ 
						= & \langle ({\bf a}^{i+j},0,0), ({\bf a}^{i+k-2},0,0), ({\bf a}^{i+k},0,0), ({\bf a}^{j+k-2},0,0), ({\bf a}^{2k-4},0,1), ({\bf a}^{2k-2},0,\delta),\\
						&({\bf a}^{j+k},0,0), ({\bf a}^{2k-2},0,\delta), ({\bf a}^{2k},1,\delta^2):~0\leq i,j\leq k-3 \rangle \\ 
						= & \langle ({\bf a}^{u},0,0), ({\bf a}^{2k-4},0,1), ({\bf a}^{2k-2},0,\delta), ({\bf a}^{2k},1,\delta^2):
						~u=0,1,2,\cdots,2k-5,2k-3 \rangle. 
					\end{split}  	  
				\end{align}
				Next, we can prove that $({\bf a}^{2k-4},0,1)$, $({\bf a}^{2k-2},0,\delta)$, $({\bf a}^{2k},1,\delta^2)$ and $({\bf a}^{u},0,0)$ are linearly independent, where $u=0,1,2,\cdots,2k-5,2k-3$, ${\bf a}^z=(\alpha_1^z, \alpha_2^z,\cdots,\alpha_n^z)$, for any nonnegative integer $z$.
				The matrix composed of these $2k$ vectors denoted by $A_{2k}$ is 
				\begin{equation*}\label{eq.matrixA2k}
					A_{2k}=\left(\begin{array}{ccccccc} 
						\alpha_1^0 & \alpha_2^0 & \cdots  & \alpha_n^0 &0 &0\\ 
						\alpha_1^1 & \alpha_2^1 & \cdots & \alpha_{n}^1 &0 &0\\
						\vdots & \vdots & \ddots & \vdots & \vdots & \vdots \\ 
						\alpha_1^{2k-5} & \alpha_2^{2k-5} & \cdots & \alpha_{n}^{2k-5} &0 &0\\
						\alpha_1^{2k-4} & \alpha_2^{2k-4} & \cdots & \alpha_{n}^{2k-4} &0 &1\\
						\alpha_1^{2k-3} & \alpha_2^{2k-3} & \cdots & \alpha_{n}^{2k-3} &0 &0\\
						\alpha_1^{2k-2} & \alpha_2^{2k-2} & \cdots & \alpha_{n}^{2k-2} &0 &\delta\\
						\alpha_1^{2k} & \alpha_2^{2k} & \cdots & \alpha_{n}^{2k} &1 &\delta\end{array}\right)_{{2k}\times{(n+2)}}.
				\end{equation*}
				Since $3\leq k<\frac{n+2}{2}$, then there exists a ${2k}\times {2k}$ submatrix $A_{2k}'$ consisting of the first $2k-1$ columns and the $(n+1)$-th column with the form of \begin{equation*}\label{eq.matrixA2k1}
					\left(\begin{array}{ccccccc} 
						\alpha_1^0 & \alpha_2^0 & \cdots  & \alpha_{2k-1}^0 &0 \\ 
						\alpha_1^1 & \alpha_2^1 & \cdots & \alpha_{2k-1}^1 &0 \\
						\vdots & \vdots & \ddots & \vdots & \vdots  \\ 
						\alpha_1^{2k-5} & \alpha_2^{2k-5} & \cdots & \alpha_{2k-1}^{2k-5} &0  \\
						\alpha_1^{2k-4} & \alpha_2^{2k-4} & \cdots & \alpha_{2k-1}^{2k-4} &0  \\
						\alpha_1^{2k-3} & \alpha_2^{2k-3} & \cdots & \alpha_{2k-1}^{2k-3} &0  \\
						\alpha_1^{2k-2} & \alpha_2^{2k-2} & \cdots & \alpha_{2k-1}^{2k-2} &0  \\
						\alpha_1^{2k} & \alpha_2^{2k} & \cdots & \alpha_{2k-1}^{2k} &1 
						\end{array}\right)_{{2k}\times{2k}}.
				\end{equation*}
				The submatrix $A_{2k}'$ is of full rank, which can implied that the rank of $A_{2k}$ is $2k$.

				So it can be concluded that $({\bf a}^{2k-4},0,1)$, $({\bf a}^{2k-2},0,\delta)$, $({\bf a}^{2k},1,\delta^2)$ and $({\bf a}^{u},0,0)$ are linearly independent, where $u=0,1,2,\cdots,2k-5,2k-3$.
				It can be further implied that the dimension of $\C_k({\bf A},{\bf 1})^2$ is $2k$.
				
				Thus $\C_k({\bf A}, {\bf 1})$ is non-GRS.  
				
				\item [] {\textbf{Case 2.}}  If $k\geq \frac{n+2}{2}$, we can firstly determine the non-GRS property of $\C_k({\bf A}, {\bf 1})^{\perp}$ from Lemma \ref{lem.GRS square distance}.
				
				Denote ${\bf h}^r=(u_1\alpha_1^r, u_2\alpha_2^r,\cdots,u_n\alpha_n^r)$ for any $0\leq r\leq n-k-2$, and ${\bf u}=(u_1,u_2,\cdots,u_n)$. 
				
				Then we have $$\C_k({\bf A}, {\bf 1})^{\perp}= \langle ({\bf h}^r, 0, 0), ({\bf h}^{n-k-1}, -1, 0), ({\bf h}^{n-k}, a, 0), ({\bf h}^{n-k+1}, b, -1) \rangle,$$
				where $a=-\sum\limits_{i=1}^{n}\alpha_i \mbox{ and } b=\delta+\sum\limits_{1\leq i<j\leq n}\alpha_i\alpha_j-(\sum\limits_{i=1}^{n}\alpha_i)^2.$
				By definitions and Theorem \ref{th.parity-check matrix}, we have 
				\begin{align}\label{eq. square dual}
					\small
					\begin{split}
						 (\C_k({\bf A}, {\bf 1})^{\perp})^2=&\C_k({\bf A}, {\bf 1})^{\perp}\star \C_k({\bf A}, {\bf 1})^{\perp}\\ 
						= &  \langle ({\bf h}^r,0,0)\star ({\bf h}^w,0,0),~ ({\bf h}^r,0,0)\star ({\bf h}^{n-k-1},-1,0),~ ({\bf h}^r,0,0)\star ({\bf h}^{n-k},a,0), \\
						& ({\bf h}^r,0,0)\star ({\bf h}^{n-k+1},b,-1),~ ({\bf h}^{n-k-1},-1,0)\star ({\bf h}^w,0,0),~ ({\bf h}^{n-k-1},-1,0)\star ({\bf h}^{n-k-1},-1,0), \\ 
						& ({\bf h}^{n-k-1},-1,0)\star ({\bf h}^{n-k},a,0),~ ({\bf h}^{n-k-1},-1,0)\star ({\bf h}^{n-k+1},b,-1),~ ({\bf h}^{n-k},a,0)\star ({\bf h}^w,0,0),\\
						& ({\bf h}^{n-k},a,0)\star ({\bf h}^{n-k-1},-1,0),({\bf h}^{n-k},a,0)\star ({\bf h}^{n-k},a,0),~({\bf h}^{n-k},a,0)\star ({\bf h}^{n-k+1},b,-1), \\
						& ({\bf h}^{n-k+1},b,-1)\star ({\bf h}^w,0,0),~ ({\bf h}^{n-k+1},b,-1)\star ({\bf h}^{n-k-1},-1,0),  \\    
						& ({\bf h}^{n-k+1},b,-1)\star ({\bf h}^{n-k},a,0), ({\bf h}^{n-k+1},b,-1)\star ({\bf h}^{n-k+1},b,-1):~0\leq r,w\leq n-k-2 \rangle \\
						= & \langle ({\bf u}\star {\bf h}^z,0,0),~({\bf u}\star {\bf h}^{2n-2k-2},1,0),~ ({\bf u}\star {\bf h}^{2n-2k-1},-a,0),~ ({\bf u}\star {\bf h}^{2n-2k},-b,0), \\
						& ({\bf u}\star {\bf h}^{2n-2k},a^2,0),~({\bf u}\star {\bf h}^{2n-2k+1},-b,0),~ ({\bf u}\star {\bf h}^{2n-2k+1},ab,0),\\
						& ({\bf u}\star {\bf h}^{2n-2k+2},b^2,1):~ 0\leq z\leq 2n-2k-2\rangle \\
						= & \langle ({\bf u}\star {\bf h}^z,0,0),~ (\underbrace{0,0,\cdots,0}_n,1,0),~ ({\bf u}\star {\bf h}^{2n-2k+2},0,1):~ 0\leq z\leq 2n-2k+1 \rangle.
					\end{split}
				\end{align}	
				Since  $k\geq \frac{n+2}{2}$, then $2n-2k-1\leq n-3$.    
				By Equation (\ref{eq. square dual}), $$d((\C_k({\bf A}, {\bf 1})^{\perp})^2)=1<2.$$
				Thus $\C_k({\bf A}, {\bf 1})^{\perp}$ is non-GRS.
				
				Recall that the dual code of a GRS code is still a GRS code. Therefore, $\C_k({\bf A}, {\bf 1})$ is non-GRS.
			\end{enumerate}
			
			This completes the proof.
		\end{proof}

		\section{NMDS codes $\C_k({\bf A},{\bf v})$}\label{sec5}

		In this section, our focus will be on searching for some necessary and sufficient conditions under which the code $\C_{k}({\bf A}, {\bf v})$ is NMDS.
		
		\subsection{When is the code $\C_k({\bf A},{\bf v})^\perp$ AMDS?}
		In this subsection, we look for necessary and sufficient conditions under which the dual code $\C_{k}({\bf A}, {\bf v})^\perp$  is AMDS. Since $G_k$ is a parity-check matrix of $\C_{k}({\bf A},{\bf v})^\perp$, we have the following results.
		
		\begin{theorem}\label{th.dAMDS condition}
			The code $\C_k({\bf A},{\bf v})^\perp$ is AMDS if and only if one of the following conditions holds:
			\begin{enumerate}
				\item [\rm (1)] There exists a subset $L\subseteq \{1,\cdots, n\}$ with size $k$ such that $$\sum_{i\in L}\alpha_i=0.$$
				
				\item [\rm (2)] There exists a subset $M\subseteq \{1,2,\cdots, n\}$ with size $k-1$ such that $$ \delta = (\sum_{ i\in M}\alpha_i)^2-\sum_{ i<j\in M}\alpha_i\alpha_j.$$
			\end{enumerate}
		\end{theorem}
		
		\begin{proof}
			The code $\C_k({\bf A},{\bf v})^\perp$ is  AMDS if and only if 
			the followings hold:
			\begin{itemize}
				\item  Any $k-1$ columns from $G_k$ in Equation \eqref{eq.generator matrix} are linearly independent.
				
				\item  There exists $k$ linearly dependent columns in $G_k$  in Equation \eqref{eq.generator matrix}.
			\end{itemize}
			
			
			Let ${\bf d}_1=(0,\cdots, 0,1)^T$ and ${\bf d}_2=(0,\cdots,0,1, \delta)^T$. Since any $k-1$ columns from $G_k$ in Equation \eqref{eq.generator matrix} are linearly independent if and only if the rank of each $k\times(k-1)$ submatrix of $G_k$ is $k-1$. So we need to calculate the conditions required for the rank of each $k\times(k-1)$ submatrix of $G_k$ to be equal to $k-1$. Then we can divide the proof into the following cases.
				\begin{enumerate}
				\item [] {\textbf{Case 1.}}  Assume that the $k\times(k-1)$ submatrix of $G_k$ consisting of $k-1$ columns of $G_k$ does not contain ${\bf d}_1$ and ${\bf d}_2$ . 
			    Suppose that the submatrix contains only any $k-1$ columns of the first $n$ columns from $G_k$ and note $I=\{i_1,i_2,\cdots,i_{k-1}\}\subseteq\{1,2,\cdots,n\}$.  We need to determine the necessary and sufficient conditions for the rank of $k\times(k-1)$ submatrix $D_1$ of $G_k$ to be $k-1$, where 
					\begin{equation*}\label{detD_1}
						D_1=\left(\begin{array}{ccccc}
							1 &1 &\cdots  &1\\ 
							\alpha_{i_1} &\alpha_{i_2} &\cdots & \alpha_{i_{k-1}}\\ 
							\vdots &\vdots &\ddots &\vdots \\ 
							\alpha_{i_1}^{k-3} &\alpha_{i_2}^{k-3} &\cdots & \alpha_{i_{k-1}}^{k-3}\\ 
							\alpha_{i_1}^{k-2} &\alpha_{i_2}^{k-2} &\cdots & \alpha_{i_{k-1}}^{k-2}\\
							\alpha_{i_1}^{k} &\alpha_{i_2}^{k} &\cdots & \alpha_{i_{k-1}}^{k}\\
						\end{array}\right)_{k\times(k-1)} .
					\end{equation*} 
					It is easy to see that any $k-1$ rows of $D_1$ is nonsingular and then the rank of matrix $D_1$ is $k-1$, equivalently, the $k-1$ columns of matrix $D_1$ are linearly independent.
				\item [] {\textbf{Case 2.}} Assume that the $k\times(k-1)$ submatrix of $G_k$ contains only ${\bf d}_1$ and other $k-2$ columns of $G_k$.
					Suppose that the submatrix contains only ${\bf d}_1$ and any $k-2$ columns of the first $n$ columns from $G_k$ and note $J=\{i_1,i_2,\cdots,i_{k-2}\}\subseteq\{1,2,\cdots,n\}$. We need to determine the necessary and sufficient conditions for the rank of $k\times(k-1)$ submatrix $D_2$ of $G_k$ to be $k-1$, where 
					\begin{equation*}\label{detD_2}
						D_2=\left(\begin{array}{ccccc}
							1 &1 &\cdots  &1 &0\\ 
							\alpha_{i_1} &\alpha_{i_2} &\cdots & \alpha_{i_{k-2}}&0\\ 
							\vdots &\vdots &\ddots &\vdots &\vdots \\ 
							\alpha_{i_1}^{k-3} &\alpha_{i_2}^{k-3} &\cdots & \alpha_{i_{k-2}}^{k-3}&0\\ 
							\alpha_{i_1}^{k-2} &\alpha_{i_2}^{k-2} &\cdots & \alpha_{i_{k-2}}^{k-2}&0\\
							\alpha_{i_1}^{k} &\alpha_{i_2}^{k} &\cdots & \alpha_{i_{k-2}}^{k}&1\\
						\end{array}\right)_{k\times(k-1)}.
					\end{equation*} 
					Deleting the penultimate row of the matrix $D_2$, then results in a $(k-1)\times(k-1)$ square matrix
					\begin{equation*}\label{detD_21}
						D_2'=\left(\begin{array}{ccccc}
							1 &1 &\cdots  &1 &0\\ 
							\alpha_{i_1} &\alpha_{i_2} &\cdots & \alpha_{i_{k-2}}&0\\ 
							\vdots &\vdots &\ddots &\vdots &\vdots \\ 
							\alpha_{i_1}^{k-3} &\alpha_{i_2}^{k-3} &\cdots & \alpha_{i_{k-2}}^{k-3}&0\\ 
							\alpha_{i_1}^{k} &\alpha_{i_2}^{k} &\cdots & \alpha_{i_{k-2}}^{k}&1\\
						\end{array}\right)_{(k-1)\times(k-1)}.
					\end{equation*} 
					It is easy to see that 
					$$\mbox{det}(D_2')=\prod_{i<j\in J}(\alpha_j-\alpha_i)\neq 0.$$
					Thus $D_2'$ is nonsingular and then the rank of matrix $D_2$ is $k-1$, equivalently, the $k-1$ columns of matrix $D_2$ are linearly independent.
				\item [] {\textbf{Case 3.}} Assume that the $k\times(k-1)$ submatrix of $G_k$ contains only ${\bf d}_2$ and other $k-2$ columns of $G_k$.
					Suppose that the submatrix contains only ${\bf d}_2$ and any $k-2$ columns of the first $n$ columns from $G_k$ and note $J=\{i_1,i_2,\cdots,i_{k-2}\}\subseteq\{1,2,\cdots,n\}$. We need to determine the necessary and sufficient conditions for the rank of $k\times(k-1)$ submatrix $D_3$ of $G_k$ to be $k-1$, where 
					\begin{equation*}\label{detD_3}
						D_3=\left(\begin{array}{ccccc}
							1 &1 &\cdots  &1 &0\\ 
							\alpha_{i_1} &\alpha_{i_2} &\cdots & \alpha_{i_{k-2}}&0\\ 
							\vdots &\vdots &\ddots &\vdots &\vdots \\ 
							\alpha_{i_1}^{k-3} &\alpha_{i_2}^{k-3} &\cdots & \alpha_{i_{k-2}}^{k-3}&0\\ 
							\alpha_{i_1}^{k-2} &\alpha_{i_2}^{k-2} &\cdots & \alpha_{i_{k-2}}^{k-2}&1\\
							\alpha_{i_1}^{k} &\alpha_{i_2}^{k} &\cdots & \alpha_{i_{k-2}}^{k}&\delta\\
						\end{array}\right)_{k\times(k-1)}.
					\end{equation*} 
					It is easy to see that the first $k-1$ rows of $D_3$ is nonsingular, and then the rank of matrix $D_3$ is $k-1$, equivalently, the $k-1$ columns of matrix $D_3$ are linearly independent.
				\item [] {\textbf{Case 4.}} Assume that the $k\times(k-1)$ submatrix of $G_k$ contains ${\bf d}_1$, ${\bf d}_2$ and other $k-3$ columns of $G_k$.
					Suppose that the submatrix contains ${\bf d}_1$, ${\bf d}_2$ and any $k-3$ columns of the first $n$ columns from $G_k$ and note $P=\{i_1,i_2,\cdots,i_{k-3}\}\subseteq\{1,2,\cdots,n\}$. We need to determine the necessary and sufficient conditions for the rank of $k\times(k-1)$ submatrix $D_4$ of $G_k$ to be $k-1$, where 
					\begin{equation*}\label{detD_4}
						D_4=\left(\begin{array}{cccccccc}
							1 &1 &\cdots  &1 &0 &0\\ 
							\alpha_{i_1} &\alpha_{i_2} &\cdots & \alpha_{i_{k-3}} &0 &0\\ 
							\vdots &\vdots &\ddots &\vdots &\vdots &\vdots\\ 
							\alpha_{i_1}^{k-3} &\alpha_{i_2}^{k-3} &\cdots & \alpha_{i_{k-3}}^{k-3}&0 &0\\ 
							\alpha_{i_1}^{k-2} &\alpha_{i_2}^{k-2} &\cdots & \alpha_{i_{k-3}}^{k-2} &0 &1\\
							\alpha_{i_1}^{k} &\alpha_{i_2}^{k} &\cdots & \alpha_{i_{k-3}}^{k}&1 &\delta\\
						\end{array}\right)_{k\times(k-1)}.
					\end{equation*}
					\begin{enumerate}
						\item []{\textbf{Subcase 4.1}} When $\alpha_{i_j}\neq 0$, for any $i_j\in P$, then we can delete the first row of the matrix $D_4$, then results in a $(k-1)\times(k-1)$ square matrix
						\begin{equation*}\label{detD_41}
							D_4'=\left(\begin{array}{ccccccccc}
								\alpha_{i_1} &\alpha_{i_2} &\cdots & \alpha_{i_{k-3}} &0 &0\\ 
								\vdots &\vdots &\ddots &\vdots &\vdots &\vdots\\ 
								\alpha_{i_1}^{k-3} &\alpha_{i_2}^{k-3} &\cdots & \alpha_{i_{k-3}}^{k-3}&0 &0\\ 
								\alpha_{i_1}^{k-2} &\alpha_{i_2}^{k-2} &\cdots & \alpha_{i_{k-3}}^{k-2} &0 &1\\
								\alpha_{i_1}^{k} &\alpha_{i_2}^{k} &\cdots & \alpha_{i_{k-3}}^{k}&1 &\delta\\
							\end{array}\right)_{(k-1)\times(k-1)}.
						\end{equation*} 
					It is easy to see that 
						\begin{equation*}
							\mbox{det}(D_4')=-\left|\begin{array}{ccccccccc}
							\alpha_{i_1} &\alpha_{i_2} &\cdots & \alpha_{i_{k-3}} &0 \\ 
							\vdots &\vdots &\ddots &\vdots &\vdots \\ 
							\alpha_{i_1}^{k-3} &\alpha_{i_2}^{k-3} &\cdots & \alpha_{i_{k-3}}^{k-3} &0\\ 
							\alpha_{i_1}^{k-2} &\alpha_{i_2}^{k-2} &\cdots & \alpha_{i_{k-3}}^{k-2}  &1\\
						\end{array}\right|\\
						=-\prod_{i\in P}\alpha_i\cdot\prod_{i<j\in P}(\alpha_j-\alpha_i)\neq 0.
						\end{equation*}
					Thus, the rank of matrix $D_4$ is $k-1$.
						\item []{\textbf{Subcase 4.2}} When there exists $i_j\in P$ such that $\alpha_{i_j}=0$,  then we can delete the second row of the matrix $D_4$, then results in a $(k-1)\times(k-1)$ square matrix
						\begin{equation*}\label{detD_42}
							D_4''=\left(\begin{array}{ccccccccc}
								1 &1 &\cdots &1 &\cdots& 1 & 0&0\\
								\alpha_1^2 &\alpha_2^2 &\cdots & 0 &\cdots & \alpha_{k-3}^2 & 0&0 \\ 
								\vdots &\vdots &\ddots &\vdots & \ddots & \vdots& \vdots &\vdots\\ 
								\alpha_1^{k-3} &\alpha_2^{k-3} &\cdots& 0 &\cdots & \alpha_{k-3}^{k-3} & 0 &0 \\
								\alpha_1^{k-2} &\alpha_2^{k-2} &\cdots & 0 &\cdots  & \alpha_{k-3}^{k-2} & 0 &1\\
								\alpha_1^{k} &\alpha_2^{k} &\cdots & 0 &\cdots  & \alpha_{k-3}^{k} & 1 & \delta\\
							\end{array}\right)_{(k-1)\times(k-1)}.
						\end{equation*}		
				It is easy to see that
						\begin{align*}
							\mbox{det}(D_4'')
							&=(-1)^{1+i_j}\left|\begin{array}{ccccccccc}
							\alpha_1^2 &\alpha_2^2 &\cdots& \alpha_{i_{j-1}}^2 & \alpha_{i_{j+1}}^2 &\cdots & \alpha_{k-3}^2 & 0 &0\\ 
							\vdots &\vdots &\ddots &\vdots & \vdots& \ddots & \vdots & \vdots &\vdots\\ 
							\alpha_1^{k-3} &\alpha_2^{k-3} &\cdots& \alpha_{i_{j-1}}^{k-3} & \alpha_{i_{j+1}}^{k-3} &\cdots & \alpha_{k-3}^{k-3} & 0 &0 \\
							\alpha_1^{k-2} &\alpha_2^{k-2} &\cdots & \alpha_{i_{j-1}}^{k-2} & \alpha_{i_{j+1}}^{k-2} &\cdots  & \alpha_{k-3}^{k-2} &0& 1\\
							\alpha_1^{k} &\alpha_2^{k} &\cdots & \alpha_{i_{j-1}}^{k} & \alpha_{i_{j+1}}^{k} &\cdots  & \alpha_{k-3}^{k} & 1 &\delta\\
						\end{array}\right| \\
						&=(-1)^{1+i_j+k-2+k-3}\left|\begin{array}{ccccccccc}
							\alpha_1^2 &\alpha_2^2 &\cdots& \alpha_{i_{j-1}}^2 & \alpha_{i_{j+1}}^2 &\cdots & \alpha_{k-3}^2 & 0 \\ 
							\vdots &\vdots &\ddots &\vdots & \vdots& \ddots & \vdots  &\vdots\\ 
							\alpha_1^{k-3} &\alpha_2^{k-3} &\cdots& \alpha_{i_{j-1}}^{k-3} & \alpha_{i_{j+1}}^{k-3} &\cdots & \alpha_{k-3}^{k-3}  &0 \\
							\alpha_1^{k-2} &\alpha_2^{k-2} &\cdots & \alpha_{i_{j-1}}^{k-2} & \alpha_{i_{j+1}}^{k-2} &\cdots  & \alpha_{k-3}^{k-2}& 1\\
						\end{array}\right| \\
					&=\left|\begin{array}{ccccccccc}
							\alpha_1^2 &\alpha_2^2 &\cdots& \alpha_{i_{j-1}}^2 & \alpha_{i_{j+1}}^2 &\cdots & \alpha_{k-3}^2 \\ 
							\vdots &\vdots &\ddots &\vdots & \vdots& \ddots & \vdots \\ 
							\alpha_1^{k-3} &\alpha_2^{k-3} &\cdots& \alpha_{i_{j-1}}^{k-3} & \alpha_{i_{j+1}}^{k-3} &\cdots & \alpha_{k-3}^{k-3} \\
						\end{array}\right| \\
					&=\prod\limits_{ i<j\in P,i,j\neq i_j}(\alpha_j-\alpha_i)\neq0.
						\end{align*}					
					Thus, the rank of matrix $D_4$ is $k-1$.
					Therefore, the rank of matrix $D_4$ is $k-1$, equivalently, the $k-1$ columns of matrix $D_4$ are linearly independent.
					\end{enumerate}
			\end{enumerate} 
			
			In a word, any $k-1$ columns  from $G_k$ in Equation \eqref{eq.generator matrix} are linearly independent.					
			Thus, the dual code $\C_k({\bf A}, {\bf v})^{\perp}$ has parameters $[n+2, n+2-k, \ge k]$ in this case. Then the desired result follows from Theorem \ref{th.MDS condition}.
						
			We have finished the whole proof.
		\end{proof}
			
		\subsection{When is the code $\C_k({\bf A},{\bf v})$ AMDS?}
		Based on Theorem \ref{th.parity-check matrix} and $G_k$ given in Equation \eqref{eq.generator matrix} is a generator matrix of $\C_k({\bf A}, {\bf v})$, we obtain the following theorem.
		
		\begin{theorem}\label{th.AMDS condition}  
			The code $\C_k({\bf A}, {\bf v})$ is AMDS if and only if one of the followings  holds:
			\begin{itemize}
				\item [\rm (1)] For any subset $I\subseteq \{1,2,\cdots,n\}$ with size $k+1$, there exists a  subset $J\subseteq I$ with size $k$, such that $$\sum\limits_{i\in J}\alpha_i\neq 0,$$ for any subset $I'\subseteq \{1,2,\cdots,n\}$ with size $k$, there exists a subset $J'\subseteq I'$ with size $k-1$, such that $$ \delta-(\sum\limits_{i\in J'}\alpha_i)^2+\sum\limits_{ i<j\in J'}\alpha_i\alpha_j\neq 0,$$
				and there exists a subset $L\subseteq \{1,\cdots, n\}$ with size $k$ such that $$\sum_{i\in L}\alpha_i=0,$$ 
				
				\item [\rm (2)] For any subset $I\subseteq \{1,2,\cdots,n\}$ with size $k+1$, there exists a  subset $J\subseteq I$ with size $k$, such that $$\sum\limits_{i\in J}\alpha_i\neq 0,$$ for any subset $I'\subseteq \{1,2,\cdots,n\}$ with size $k$, there exists a subset $J'\subseteq I'$ with size $k-1$, such that $$ \delta-(\sum\limits_{i\in J'}\alpha_i)^2+\sum\limits_{ i<j\in J'}\alpha_i\alpha_j\neq 0,$$
				and there exists a subset $M\subseteq \{1,2,\cdots, n\}$ with size $k-1$ such that $$ \delta = (\sum_{ i\in M}\alpha_i)^2-\sum_{ i<j\in M}\alpha_i\alpha_j.$$ 
			\end{itemize}
					
		\end{theorem}
		
		\begin{proof} 
			The code $\C_k({\bf A},{\bf v})$ is AMDS if and only if 
				the followings hold: 
			\begin{itemize}
				\item There exists $k$ linearly dependent columns in $G_k$  in Equation \eqref{eq.generator matrix}. 
				
				\item Any $k+1$ columns from $G_{k}$ in Equation \eqref{eq.generator matrix} are of full rank. 
			\end{itemize}
%
			
			Let ${\bf d}_1=(0,0,\cdots, 0,1)^T \mbox{ and } {\bf d}_2=(0,0,\cdots,0,1,\delta)^T$.  Since any $k+1$ columns from $G_{k}$ in Equation \eqref{eq.generator matrix} are linearly independent if and only if the rank of each $k\times(k+1)$ submatrix of $G_{k}$ is $k$. So we need to calculate the conditions required for the rank of each $k\times(k+1)$ submatrix of $G_{k}$ to be equal to $k$. Then we can divide the proof into the following cases. 
		
			\begin{itemize}
				\item [] {\textbf{Case 1.}}   Assume that the $k\times(k+1)$ submatrix of $G_{k}$ consisting of $k+1$ columns of the first $n$ columns from $G_{k}$ does not contain ${\bf d}_1$ and ${\bf d}_2$. 
				Suppose that the submatrix contains any $k+1$ columns of the first $n$ columns from $G_{k}$ and note $I=\{i_1,i_2,\cdots ,i_{k+1}\}\subseteq \{1,2,\cdots,n\}$. We need to determine the necessary and sufficient conditions for the rank of $k\times(k+1)$ submatrix $E_1$, where
				\begin{equation*}\label{detE_1}
					E_1=\left(\begin{array}{cccccccc}
						1 &1 &\cdots  &1 &1\\ 
						\alpha_{i_1} &\alpha_{i_2} &\cdots & \alpha_{i_k} & \alpha_{i_{k+1}}\\ 
						\vdots &\vdots &\ddots &\vdots &\vdots \\ 
						\alpha_{i_1}^{k-2} &\alpha_{i_2}^{k-2} &\cdots & \alpha_{i_k}^{k-2}& \alpha_{i_{k+1}}^{k-2}\\ 
						\alpha_{i_1}^{k} &\alpha_{i_2}^{k} &\cdots & \alpha_{i_k}^{k}& \alpha_{i_{k+1}}^k
					\end{array}\right)_{k\times (k+1)}.
				\end{equation*} 
				
				Deleting any column of the matrix $E_1$, let $J=\{j_1,j_2,\cdots,j_k\}\subseteq I$, then results in a $k\times k$ square matrix 
				\begin{equation*}\label{detE_11}
					E_1'=\left(\begin{array}{cccccccc}
						1 &1 &\cdots  &1\\ 
						\alpha_{j_1} &\alpha_{j_2} &\cdots & \alpha_{j_{k}}\\ 
						\vdots &\vdots &\ddots &\vdots \\ 
						\alpha_{j_1}^{k-2} &\alpha_{j_2}^{k-2} &\cdots & \alpha_{j_{k}}^{k-2}\\  
						\alpha_{j_1}^{k} &\alpha_{j_2}^{k} &\cdots & \alpha_{j_{k}}^{k}\\
					\end{array}\right)_{k\times k}.
				\end{equation*}
				
				Then 
				$$\mbox{det}(E_1')=\sum_{i\in J}\alpha_i\prod_{ i<j\in J}(\alpha_j-\alpha_i).$$
				
				Therefore, $\mbox{det}(E_1')\neq 0$ if and only if  for any subset $I\subseteq \{1,2,\cdots,n\}$ with size $k+1$, there exists a subset $J\subseteq I$ with size $k$, such that $\sum\limits_{i\in J}\alpha_i\neq 0$, and then the rank of matrix $E_1$ is $k$, equivalently, the matrix $E_1$ are of full rank if and only if for any subset $I\subseteq \{1,2,\cdots,n\}$ with size $k+1$, there exists a subset $J\subseteq I$ with size $k$, such that $\sum\limits_{i\in J}\alpha_i\neq 0$.

				\item [] {\textbf{Case 2.}}  Assume that the $k\times(k+1)$ submatrix of $G_{k}$ consisting of ${\bf d}_1$ and any $k$ columns of the first $n$ columns from $G_{k}$.
				Suppose that the submatrix contains ${\bf d}_1$ and any $k$ columns of the first $n$ columns from $G_{k}$ and note $K=\{i_1,i_2,\cdots ,i_{k}\}\subseteq \{1,2,\cdots,n\}$. We need to determine the necessary and sufficient conditions for the rank of $k\times(k+1)$ submatrix $E_2$, where
					\begin{equation*}\label{detE_2}
						E_2=\left(\begin{array}{cccccccc}
							1 &1 &\cdots &1 &0\\ 
							\alpha_{i_1} &\alpha_{i_2} &\cdots & \alpha_{i_k} & 0\\ 
							\vdots &\vdots &\ddots &\vdots &\vdots \\ 
							\alpha_{i_1}^{k-2} &\alpha_{i_2}^{k-2} &\cdots & \alpha_{i_k}^{k-2}& 0\\ 
							\alpha_{i_1}^{k} &\alpha_{i_2}^{k} &\cdots & \alpha_{i_k}^{k}& 1
						\end{array}\right)_{k\times (k+1)}.
				\end{equation*} 
				It is easy to see that the $k$ rows of the matrix $E_2$ are linearly independent, thus the matrix $E_2$ are of full rank.
				
				\item [] {\textbf{Case 3.}}  Assume that the $k\times(k+1)$ submatrix of $G_{k}$ consisting of ${\bf d}_2$ and any $k$ columns of the first $n$ columns from $G_{k}$.
				Suppose that the submatrix contains ${\bf d}_2$ and any $k$ columns of the first $n$ columns from $G_{k}$ and note $I'=\{i_1,i_2,\cdots ,i_{k}\}\subseteq \{1,2,\cdots,n\}$. We need to determine the necessary and sufficient conditions for the rank of $k\times(k+1)$ submatrix $E_3$, where
					\begin{equation*}\label{detE_3}
						E_3=\left(\begin{array}{cccccccc}
							1 &1 &\cdots &1&0\\ 
							\alpha_{i_1} &\alpha_{i_2} &\cdots & \alpha_{i_k} & 0\\ 
							\vdots &\vdots &\ddots &\vdots &\vdots \\ 
							\alpha_{i_1}^{k-2} &\alpha_{i_2}^{k-2} &\cdots & \alpha_{i_k}^{k-2}& 1\\ 
							\alpha_{i_1}^{k} &\alpha_{i_2}^{k} &\cdots & \alpha_{i_k}^{k}& \delta
						\end{array}\right)_{k\times (k+1)}.
				\end{equation*} 
				
				Deleting any column of the first $k$ columns of the matrix $E_1$, let $J'=\{j_1,j_2,\cdots,j_{k-1}\}\subseteq I'$, then results in a $k\times k$ square matrix 
					\begin{equation*}\label{detE_31}
						E_3'=\left(\begin{array}{cccccccc}
							1 &1 &\cdots  &1&0\\ 
							\alpha_{j_1} &\alpha_{j_2} &\cdots & \alpha_{j_{k-1}}&0\\ 
							\vdots &\vdots &\ddots &\vdots &\vdots \\ 
							\alpha_{j_1}^{k-2} &\alpha_{j_2}^{k-2} &\cdots & \alpha_{j_{k-1}}^{k-2}&1\\  
							\alpha_{j_1}^{k} &\alpha_{j_2}^{k} &\cdots & \alpha_{j_{k-1}}^{k}&\delta\\
						\end{array}\right)_{k\times k}.
				\end{equation*}
			    Then 
					\begin{align*}
						\mbox{det}(E_3')&=
						\delta\cdot \left|\begin{array}{cccccccc}
							1 &1 &\cdots  &1 \\ 
							\alpha_{j_1} &\alpha_{j_2} &\cdots & \alpha_{j_{k-1}} \\ 
							\vdots &\vdots &\ddots  &\vdots \\ 
							\alpha_{j_1}^{k-2} &\alpha_{j_2}^{k-2} &\cdots & \alpha_{j_{k-1}}^{k-2} \\  
						\end{array}\right|-
						\left|\begin{array}{cccccccc}
							1 &1 &\cdots  &1\\ 
							\alpha_{j_1} &\alpha_{j_2} &\cdots & \alpha_{j_{k-1}}\\ 
							\vdots &\vdots &\ddots   &\vdots \\ 
							\alpha_{j_1}^{k-3} &\alpha_{j_2}^{k-3} &\cdots & \alpha_{j_{k-1}}^{k-3}\\  
							\alpha_{j_1}^{k} &\alpha_{j_2}^{k} &\cdots & \alpha_{j_{k-1}}^{k}\\
						\end{array}\right|\\
						&=\bigg(\delta-(\sum\limits_{i\in J'}\alpha_i)^2+\sum\limits_{ i<j\in J'}\alpha_i\alpha_j\bigg) \cdot \prod\limits_{ i<j\in J'}(\alpha_j-\alpha_i).
					\end{align*}
					Therefore, $\mbox{det}(E_3')\neq 0$ if and only if for any subset $I'\subseteq \{1,2,\cdots,n\}$ with size $k$, there exists a subset $J'\subseteq I'$ with size $k-1$, such that  $\delta-(\sum\limits_{i\in J'}\alpha_i)^2+\sum\limits_{ i<j\in J'}\alpha_i\alpha_j\neq 0$, and then the rank of matrix $E_3$ is $k$, equivalently, the matrix $E_3$ are of full rank if and only if for any subset $I'\subseteq \{1,2,\cdots,n\}$ with size $k$, there exists a subset $J'\subseteq I'$ with size $k-1$, such that  $\delta-(\sum\limits_{i\in J'}\alpha_i)^2+\sum\limits_{ i<j\in J'}\alpha_i\alpha_j\neq 0$. 
				
				\item [] {\textbf{Case 4.}}  Assume that the $k\times(k+1)$ submatrix of $G_{k}$ consisting of ${\bf d}_1$, ${\bf d}_2$ and any $k-1$ columns of the first $n$ columns from $G_{k}$.
					Suppose that the submatrix contains ${\bf d}_1$, ${\bf d}_2$  and any $k-1$ columns of the first $n$ columns from $G_{k}$ and note $K'=\{i_1,i_2,\cdots ,i_{k-1}\}\subseteq \{1,2,\cdots,n\}$. We need to determine the necessary and sufficient conditions for the rank of $k\times(k+1)$ submatrix $E_4$, where
					\begin{equation*}\label{detE_4}
						E_4=\left(\begin{array}{cccccccc}
						1 &1 &\cdots  &1 &0 &0\\ 
						\alpha_{i_1} &\alpha_{i_2} &\cdots & \alpha_{i_{k-1}}& 0 &0\\ 
						\vdots &\vdots &\ddots &\vdots &\vdots &\vdots\\ 
						\alpha_{i_1}^{k-3} &\alpha_{i_2}^{k-3} &\cdots & \alpha_{i_{k-1}}^{k-3} &0 &0\\ 
						\alpha_{i_1}^{k-2} &\alpha_{i_2}^{k-2} &\cdots & \alpha_{i_{k-1}}^{k-2} &0 &1\\
						\alpha_{i_1}^{k} &\alpha_{i_2}^{k} &\cdots & \alpha_{i_{k-1}}^{k} &1 &\delta
						\end{array}\right)_{k\times (k+1)}.
					\end{equation*} 
					It is easy to see that the $k$ rows of the matrix $E_4$ are linearly independent, thus the matrix $E_4$ are of full rank.
			\end{itemize}	
				
			In a word,  any $k+1$ columns  from $G_{k}$ in Equation \eqref{eq.generator matrix} are of full rank if and only if for any subset $I\subseteq \{1,2,\cdots,n\}$ with size $k+1$, there exists a  subset $J\subseteq I$ with size $k$, such that $\sum\limits_{i\in J}\alpha_i\neq 0$, and for any subset $I'\subseteq \{1,2,\cdots,n\}$ with size $k$, there exists a subset $J'\subseteq I'$ with size $k-1$, such that $ \delta-(\sum\limits_{i\in J'}\alpha_i)^2+\sum\limits_{ i<j\in J'}\alpha_i\alpha_j\neq 0$. 
				
			Thus, the code $\C_k({\bf A}, {\bf v})$ has parameters $[n+2, k, \ge n-k+2]$ in this case. Then the desired result follows from Theorem \ref{th.MDS condition}.
			Since any $k$ columns from $G_{k}$ is nonsingular if and only if $G_k$ is MDS, we can immediately get the desired result by Theorem \ref{th.MDS condition}.
			
			We finished the whole proof. 
		\end{proof}

		Combining Theorems \ref{th.MDS condition}, \ref{th.dAMDS condition} and \ref{th.AMDS condition} yields the following theorem.
		
		\begin{theorem}\label{th.NMDS condition}
			The code $\C_k({\bf A},{\bf v})$ is NMDS if and only if one of the followings holds:
			\begin{enumerate}
					\item [\rm (1)] For any subset $I\subseteq \{1,2,\cdots,n\}$ with size $k+1$, there exists a  subset $J\subseteq I$ with size $k$, such that $$\sum\limits_{i\in J}\alpha_i\neq 0,$$ for any subset $I'\subseteq \{1,2,\cdots,n\}$ with size $k$, there exists a subset $J'\subseteq I'$ with size $k-1$, such that $$ \delta-(\sum\limits_{i\in J'}\alpha_i)^2+\sum\limits_{ i<j\in J'}\alpha_i\alpha_j\neq 0,$$
					and there exists a subset $L\subseteq \{1,\cdots, n\}$ with size $k$ such that $$\sum_{i\in L}\alpha_i=0;$$
				
				\item [\rm (2)] For any subset $I\subseteq \{1,2,\cdots,n\}$ with size $k+1$, there exists a  subset $J\subseteq I$ with size $k$, such that $$\sum\limits_{i\in J}\alpha_i\neq 0,$$ for any subset $I'\subseteq \{1,2,\cdots,n\}$ with size $k$, there exists a subset $J'\subseteq I'$ with size $k-1$, such that $$ \delta-(\sum\limits_{i\in J'}\alpha_i)^2+\sum\limits_{ i<j\in J'}\alpha_i\alpha_j\neq 0,$$
					and there exists a subset $M\subseteq \{1,2,\cdots, n\}$ with size $k-1$ such that $$ \delta = (\sum_{ i\in M}\alpha_i)^2-\sum_{ i<j\in M}\alpha_i\alpha_j.$$
			\end{enumerate}
			
		\end{theorem}

		\section{Conclusions}\label{sec.concluding remarks}
		
		The main contributions of this paper are summarized as follows:
		\begin{itemize}
			\item We proved that $\C_k({\bf A}, {\bf v})$ is an extended code of the code $\C_4$ in Theorem \ref{th.extended code}.
			
			\item We gave the parameters and a parity-cheek matrix of $\C_k({\bf A}, {\bf v})$ in Theorems \ref{th.length and dimension} and  \ref{th.parity-check matrix}.

			\item We presented some sufficient and necessary conditions for $\C_k({\bf A}, {\bf v})$ to be an MDS code in Theorem \ref{th.MDS condition}.

			\item We proved that these codes are non-GRS by using the Schur method in Theorem \ref{th.non-GRS}.
			
			\item We derived and presented some sufficient and necessary conditions for both $\C_k({\bf A}, {\bf v})$ and its dual $\C_k({\bf A}, {\bf v})^{\perp}$  to be \textcolor{red}{an} AMDS codes in Theorems \ref{th.AMDS condition} and \ref{th.dAMDS condition}.
			
			\item We obtained and presented some sufficient and necessary conditions for $\C_k({\bf A}, {\bf v})$ to be \textcolor{red}{an} NMDS code in Theorem \ref{th.NMDS condition}.
			
		\end{itemize}

		\section*{Declarations}
		
		\noindent\textbf{Data availability} No data are generated or analyzed during this study.  \\
		
		\noindent\textbf{Conflict of Interest} The authors declare that there is no possible conflict of interest.

	\end{sloppypar}

\begin{thebibliography}{99}\setlength{\itemsep}{0.8mm}
			
			\bibitem{B2015}S. Ball, Finite geometry and combinatorial applications, Cambridge University Press, Cambridge, 2015. 
			
			\bibitem{Magma}W. Bosma, J. Cannon, C. Playoust, The Magma algebra system I, The user language, J. Symb. Comput. 24 (3-4) (1997) 235-265.
			
			\bibitem{BBPR2018}  P. Beelen,  M. Bossert,  S. Puchinger, J. Rosenkilde, Structural properties of twisted Reed-Solomon codes with applications to cryptography, in: Proc. IEEE Int. Symp. Inf. Theory (ISIT), 2018, pp. 946-950.
			
			\bibitem{BPR2022}  P. Beelen, S. Puchinger,  J. Rosenkilde, Twisted Reed-Solomon codes, IEEE Trans. Inf. Theory 68 (5) (2022) 3047-3061. 
		
			\bibitem{CW2023}W. Cheng, On parity-check matrices of twisted generalized Reed-Solomon codes, IEEE Trans. Inf. Theory 70 (5) (2024) 3213-3225.
			
			\bibitem{CHL2011} V.R. Cadambe, C. Huang, J. Li, Permutation code: Optimal exact-repair of a single failed node in MDS code based distributed storage systems, in: Proc. IEEE Int. Symp. Inf. Theory (ISIT), 2011, pp. 1225-1229.
			
		
			\bibitem{DL1994}  S. Dodunekov, I. Landjev,  On near-MDS codes, J. Geom., 54 (1-2) (1994) 30-43. 
			
			\bibitem{DT2020}C. Ding, C. Tang, Infinite families of near MDS codes holding $t$-designs, IEEE Trans. Inf. Theory 66 (9) (2020) 5419-5428. 
			
			\bibitem{FLL2020}X. Fang, M. Liu, J. Luo, New MDS Euclidean self-orthogonal codes, IEEE Trans. Inf. Theory 67 (1) (2020) 130-137. 
			
			\bibitem{FX2024} W. Fang, J. Xu, Deep holes of twisted Reed-Solomon codes, IEEE Int. Symp. Inf. Theory (ISIT), Athens, Greece, 2024, pp. 488-493.
			
			\bibitem{GLLS2023} G. Guo, R. Li, Y. Liu, H. Song, Duality of generalized twisted Reed-Solomon codes and Hermitian self-dual MDS or NMDS codes, Cryptogr. Commun. 15 (2) (2023) 383-395.
			
			\bibitem{GZ2023} H. Gu, J. Zhang, On twisted generalized Reed-Solomon codes with $\ell$ twists, IEEE Trans. Inf. Theory 70 (1) (2023) 145-153. 
			
			\bibitem{H1929} E.R. Heineman, Generalized Vandermonde determinants, Trans. Amer. Math. Soc. 31 (3) (1929) 464-476. 
			
			\bibitem{HF2023} D. Han, C. Fan, Roth-Lempel NMDS codes of non-elliptic-curve type, IEEE Trans. Inf. Theory 69 (9) (2023) 5670-5675.
			
			\bibitem{HL2023} B. He, Q. Liao, The properties and the error-correcting pair for lengthened GRS codes, Des. Codes Cryptogr. 92 (2023) 211-225. 
			
			
			\bibitem{HP2003} W.C. Huffman, V. Pless, Fundamentals of Error-Correcting Codes, Cambridge University Press, Cambridge, 2003.
		
			
			\bibitem{HYNL2021} D. Huang, Q. Yue, Y. Niu, X. Li, MDS or NMDS self-dual codes from twisted generalized Reed-Solomon codes, Des. Codes Cryptogr. 89 (9) (2021) 2195-2209.
			
			
			\bibitem{HZ2023} D. Han, H. Zhang, New constructions of NMDS self-dual codes, (2023).\url{https://doi.org/10.48550/arXiv.2308.01593}. 
			
			
			\bibitem{LCEL2022} G. Luo, X. Cao, M.F. Ezerman, S. Ling, Two new classes of Hermitian self-orthogonal non-GRS MDS codes and their applications, Adv. Math. Commun. 16 (4) (2022) 921-933.
			
			\bibitem{LDMTT2022} Q. Liu, C. Ding, S. Mesnager, C. Tang, V.D. Tonchev, On infinite families of narrow-sense antiprimitive BCH codes admitting $3$-transitive automorphism groups and their consequences, IEEE Trans. Inf. Theory 68 (5) (2022) 3096-3107. 
			
			\bibitem{LL2021} H. Liu, S. Liu, Construction of MDS twisted Reed-Solomon codes and LCD MDS codes, Des. Codes Cryptogr. 89 (2021) 2051-2065.
			
			
			\bibitem{LR2020} J. Lavauzelle, J. Renner, Cryptanalysis of a system based on twisted Reed-Solomon codes, Des. Codes Cryptogr. 88 (7) (2020) 1285-1300.
			
			\bibitem{LXW2008} Z. Li, L. Xing, X. Wang, Quantum generalized Reed-Solomon codes, Unified framework for quantum MDS codes. Phys. Rev. A 77 (1) (2008) 012308.
			
			\bibitem{LZ2024} Y. Li, S. Zhu, New non-GRS type MDS codes and NMDS codes, (2024). \url{https://doi.org/10.48550/arXiv.2401.04360}.
		
		
			\bibitem{MMP2013} I. M\'arquez-Corbella, E. Mart\'inez-Moro, R. Pellikaan, The non-gap sequence of a subcode of a generalized Reed-Solomon code, Des. Codes Cryptogr. 66 (2013) 317-333. 
			
			
			\bibitem{RL1989} R.M. Roth, A. Lempel, A construction of non-Reed-Solomon type MDS codes. IEEE Trans. Inf. Theory 35 (3) (1989) 655-657.
			
			\bibitem{SD2023} Z. Sun, C. Ding, The extended codes of a family of reversible MDS cyclic codes, IEEE Trans. Inf. Theory 70 (7) (2024) 4808-4822.
			
			\bibitem{SDC2023} Z. Sun, C. Ding, T. Chen, The extended codes of some linear codes, Finite Fields Appl. 96 (2024) 102401.
			
			\bibitem{ST2013} K. Sakakibara, J. Taketsugu, Application of random relaying of partitioned MDS codeword block to persistent relay CSMA over random error channels, in: Proc. $5$th Int. Congr. Ultra Modern Telecommun. Control Syst. Workshops (ICUMT), 2013, pp. 106-112.
			
			\bibitem{SV2018} D.E. Simos, Z. Varbanov, MDS codes, NMDS codes and their secret-sharing schemes, (2018). 
			\url{http://www.singacom.uva.es/~edgar/cactc2012/trabajos/CACT2012_SimosVarbanov.pdf}. 
			
			\bibitem{SYLH2022} J. Sui, Q. Yue, X. Li, D. Huang, MDS, near-MDS or $2$-MDS self-dual codes via twisted generalized Reed-Solomon codes, IEEE Trans. Inf. Theory 68 (12) (2022) 7832-7841. 
			
			\bibitem{SYS2023} J. Sui, Q. Yue, F. Sun, New constructions of self-dual codes via twisted generalized Reed-Solomon codes, Cryptogr. Commun. 15 (5) (2023) 959-978. 
			
			
			\bibitem{TR2018} A. Thomas, B. Rajan, Binary informed source codes and index codes using certain near-MDS codes, IEEE Trans. Commun. 66 (5) (2018) 2181-2190. 
			
		
			
			\bibitem{W2021} Y. Wu, Twisted Reed-Solomon codes with one-dimensional hull, IEEE Commun. Lett. 25 (2) (2021) 383-386. 
			
			
			
			\bibitem{WHL2021} Y. Wu, J.Y. Hyun, Y. Lee, New LCD MDS codes of non-Reed-Solomon type, IEEE Trans. Inf. Theory 67 (8) (2021) 5069-5078.
			
			\bibitem{WHLD2024} Y. Wu, Z. Heng, C. Li, C. Ding, More MDS codes of non-Reed-Solomon type, (2024). \url{https://doi.org/10.48550/arXiv.2401.03391}.
			
			\bibitem{ZLA2024} C. Zhu, Q. Liao, A class of double-twisted generalized Reed-Solomon codes, Finite Fields Appl. 95 (2024) 102395.
			
			\bibitem{ZLT2024} C. Zhu, Q. Liao, The (+)-extended twisted generalized Reed-Solomon code, Discret. Math. 347 (2) (2024) 113749. 
			
			\bibitem{ZWXLQY2009} Y. Zhou,  F. Wang, Y. Xin, S. Luo, S. Qing, Y. Yang, A secret sharing scheme based on near-MDS codes, in: Proc. IC-NIDC, 2009, pp. 833-836. 
			
			\bibitem{ZZT2022} J. Zhang, Z. Zhou, C. Tang, A class of twisted generalized Reed-Solomon codes, Des. Codes Cryptogr. 90 (7) (2022) 1649-1658. 
			
		\end{thebibliography}
\end{document}